\documentclass[11pt,letterpaper]{article}

\usepackage[utf8]{inputenc}
\usepackage[T1]{fontenc}
\usepackage{microtype}

\usepackage[margin=1in]{geometry}

\usepackage{amsmath,mathtools}
\usepackage{amsthm,thmtools}

\usepackage{amssymb}
\usepackage{bm,dsfont}
\usepackage{csquotes}

\usepackage{xcolor}

\usepackage{paralist}

\usepackage[small]{caption}
\usepackage{tikz}
\usepackage{pgfplots}
\pgfplotsset{compat=newest}

\usepackage{booktabs,tabularew}
\usepackage[para]{threeparttable}

\usepackage[ruled,linesnumbered,noline]{algorithm2e}

\usepackage{array}

\usepackage[
 backend=biber, style={./mynumeric}, citetracker, hyperref=true,
 maxcitenames=2, %
 sortcites,      %
 sorting=nyt,
 maxbibnames=12, %
 giveninits,     %
 date=year,      %
 proctitle=false,%
 isbn=false,     %
 url=false,      %
 doi=false,      %
 eprint=false,   %
]{biblatex}

\usepackage[colorlinks,pdftitle={The Power of Proportional Fairness for Non-Clairvoyant Scheduling under Polyhedral Constraints}]{hyperref}
\usepackage[capitalize, noabbrev]{cleveref}
\usepackage{xurl}
\hypersetup{breaklinks=true,citecolor=green!60!black,linkcolor=red!70!black,urlcolor=blue!70!black}

\usepackage{xspace}

\declaretheorem{theorem}
\declaretheorem{lemma}
\declaretheorem{observation}
\declaretheorem{corollary}
\declaretheorem[style=definition]{definition}

\DeclarePairedDelimiter\ceil{\lceil}{\rceil}
\DeclarePairedDelimiter\floor{\lfloor}{\rfloor}

\DeclareMathOperator*{\argmax}{arg\,max}
\DeclareMathOperator*{\argmin}{arg\,min}

\newcommand{\RR}{\mathbb{R}}
\newcommand{\Q}{\mathbb Q}
\newcommand{\cO}{\mathcal{O}}

\newcommand{\EX}{\mathbb{E}}

\newcommand{\ind}{\ensuremath{\mathds{1}}}
\newcommand{\bone}{\mathbf{1}}

\newcommand{\abc}[3]{$#1\mid #2\mid #3$}
\newcommand{\pmtn}{\mathrm{pmtn}}

\newcommand{\psp}{\relax\ifmmode{\mathrm{PSP}}\else PSP\xspace\fi}
\newcommand{\monpsp}{\textsc{MonPSP}\xspace}

\newcommand{\alg}{{\mathrm{ALG}}}
\newcommand{\rr}{RR\xspace}
\newcommand{\wrr}{WRR\xspace}

\newcommand{\wdeq}{WDEQ\xspace}
\newcommand{\pf}{\relax\ifmmode{\mathrm{PF}}\else PF\xspace\fi}
\newcommand{\pspt}{preemptive SPT\xspace}

\newcommand{\opt}{{\mathrm{OPT}}}
\newcommand{\optnp}{{\mathrm{OPT}^{\mathrm{np}}_{0}}}
\newcommand{\CP}{\mathrm{CP}}
\newcommand{\LP}{\mathrm{LP}}
\newcommand{\DLP}{\mathrm{DLP}}

\newcommand{\cP}{\mathcal{P}}

\newcommand{\lagrangevar}{\eta}
\newcommand{\hlagrangevar}{\hat{\lagrangevar}}
\newcommand{\hy}{\hat{y}}
\newcommand{\dualSa}{\bar\alpha}
\newcommand{\dualSb}{\bar\beta}
\newcommand{\dualVa}{\alpha}
\newcommand{\dualVb}{\beta}

\definecolor{bluegreen}{HTML}{33615E}
\definecolor{lime}{HTML}{56A778}
\definecolor{yellow}{HTML}{E5CB6D}
\definecolor{orange}{HTML}{FB8318}
\definecolor{lightred}{HTML}{DD6143}

\definecolor{job1}{HTML}{A6CEE3}
\definecolor{job2}{HTML}{1F78B4}
\definecolor{job5}{HTML}{B2DF8A}
\definecolor{job6}{HTML}{33A02C}
\definecolor{job3}{HTML}{FB9A99}
\definecolor{job4}{HTML}{E31A1C}

\newcommand{\gray}[1]{\textcolor{gray}{#1}}

\title{The Power of Proportional Fairness for Non-Clairvoyant Scheduling under Polyhedral Constraints}

\author{Sven Jäger\thanks{\texttt{sven.jaeger@math.rptu.de}, Department of Mathematics, RPTU Kaiserslautern-Landau, Germany.} \and Alexander Lindermayr\thanks{\texttt{\{linderal,nmegow\}@uni-bremen.de}, Faculty of Mathematics and Computer Science, University of Bremen, Germany.} \and Nicole Megow\footnotemark[2]}

\date{}

\bibliography{literature}

\begin{document}

\maketitle

\begin{abstract}
  The Polytope Scheduling Problem 
  (PSP) was introduced by Im, Kulkarni, and Munagala (JACM 2018) as a very general abstraction of resource allocation  
  over time and  captures many well-studied problems including classical unrelated machine scheduling, multidimensional scheduling, and broadcast scheduling. In PSP, jobs with different 
  arrival times receive processing rates that are subject to arbitrary packing constraints. 
  An elegant and well-known algorithm for instantaneous rate allocation with good fairness and efficiency properties is the Proportional Fairness algorithm (PF), which was analyzed for PSP by Im et al.
  
  We drastically improve the analysis of the PF algorithm for both the general PSP and several of its important special cases subject to the objective of minimizing the sum of weighted completion times. 
  We reduce the upper bound on the competitive ratio from 128 to 27 for general PSP and to 4 for the prominent class of monotone PSP.  
  For certain heterogeneous machine environments we even close the substantial gap to the lower bound of 2 for non-clairvoyant scheduling. Our analysis also gives the first polynomial-time improvements over the nearly 30-year-old bounds on the competitive ratio of the doubling framework by Hall, Shmoys, and Wein (SODA 1996) for clairvoyant online preemptive scheduling on unrelated machines. Somewhat surprisingly, we achieve this improvement by a non-clairvoyant algorithm, thereby demonstrating that non-clairvoyance is not a (significant) hurdle.
  
  Our improvements are based on exploiting monotonicity properties of PSP, providing tight dual fitting arguments on structured instances, and showing new additivity properties on the optimal objective value for scheduling on unrelated machines. 
  Finally, we establish new connections of PF to matching markets, and thereby provide new insights on equilibria and their computational complexity.
\end{abstract}

\thispagestyle{empty}

\newpage

\section{Introduction}

Two key challenges in the scheduling of modern cloud or general-purpose computing systems are~the arrival of jobs without prior knowledge (\emph{online arrival}) and the lack of information regarding their processing requirements (\emph{non-clairvoyance}).
We investigate online algorithms designed to cope with these uncertainties while aiming to minimize the total weighted completion time, a central quality-of-service objective. In this context, weights can be used to model priorities.
The performance of an online algorithm is evaluated using the competitive ratio, defined as the worst-case ratio between the algorithm's objective value and that of an offline optimal solution for any given~instance.

To capture a variety of \emph{preemptive} scheduling environments,
we adopt the abstract perspective of \emph{instantaneous resource allocation}. Specifically, at any time $t$, an algorithm can process a job $j$ at a \emph{rate} 
$y_j(t)$ such that the \emph{rate vector} $y(t)$ over all $n$ jobs satisfies polyhedral constraints given by a packing polytope $\cP = \{y \in \RR^n_{\ge 0} \mid B \cdot y \le \bone\}$, where $B \in \mathbb{Q}_{\geq 0}^{D \times n}$.
A job is considered complete when it has received the total required processing.
This abstract problem, termed \emph{Polytope Scheduling Problem~(\psp)}, was introduced in a seminal paper by \mbox{\textcite{ImKM18}}. It generalizes many well-studied preemptive scheduling problems, such as unrelated machine scheduling~\cite{Sku01}, multidimensional scheduling~\cite{GZH+11}, or broadcast scheduling~\cite{BansalCS08}. The power of the general \psp formulation is that it abstracts away the specific details of the scheduling environment and extracts the essence of ``packing over time''.

We consider 
a simple and well-known
mechanism, known as \emph{Proportional Fairness} (\pf), 
which dates back to Nash~\cite{nash1950bargaining} and is widely studied in the context of fair allocation and markets~\cite{Mou03,JainV10,PTV22,JalotaPQY23}.
At any time~$t$ it allocates processing rates~$y_j(t)$ to
the available jobs~$j$ such that the \emph{Weighted Nash Social Welfare}~\cite{kaneko1979nash} is maximized, where the rates are interpreted as utilities.
This has desirable properties from the perspective of \emph{fairness}, such as Pareto-efficiency and envy-freeness~\cite{varian1976two,ImKM18,TV24}. 
\Textcite{ImKM18} were the first who studied this allocation rule in the context of scheduling. 
They proved that \pf 
is a $128$\footnote{The factor stated in~\cite{ImKM18} is $64$, but the proof contained an algebraic mistake and is fixed by losing a factor $2$~\cite{ImKM18-erratum}.}-competitive non-clairvoyant algorithm for \psp, while the best-known non-clairvoyant lower bound
for this problem is only $2$. Interestingly, this lower bound 
already holds in the simple 
case of single-machine scheduling, a classical result by Motwani~et~al.~\cite{MPT94}.

A fundamental 
\psp and central problem in scheduling and resource allocation theory is the problem of scheduling \emph{unrelated machines}~\cite{LenstraST90,HSSW97,Sku01,BSS21,ImL23,Harris24,li2024approximating,BienkowskiKL21,GuptaMUX20,AnandGK12,ImKMP14}. In this problem,
denoted as \abc{R}{r_j, \pmtn}{\sum w_j C_j}, we are given~$m$ machines, and every job $j$ processes on machine $i$ at speed $s_{ij} \geq 0$. 
At any
time a job can be assigned to at most one machine and on every machine there can be at most one job assigned. These matching constraints
determine the
polytope for the this \psp. Any fractional rate vector in this polytope  can be efficiently transformed to a preemptive and \emph{migratory} schedule for one time unit.
We emphasize that both preemption~\cite{MPT94} and migration are necessary for competitive non-clairvoyant algorithms for this problem; we provide a general lower bound for non-migratory algorithms.
Important and well-studied special cases of unrelated machines 
are \emph{related} machines (in short \abc{Q}{r_j, \pmtn}{\sum w_jC_j}), where on every machine $i$ the speed $s_i = s_{ij}$ is the same for every job~$j$, and \emph{restricted assignment} (in short \abc{R}{r_j, \pmtn,s_{ij} \in \{0,1\}}{\sum w_jC_j}), where all $s_{ij} \in \{0,1\}$.

Little is known about scheduling jobs with online arrival on unrelated machines with preemption and migration.
With respect to non-clairvoyant algorithms, the bound of $128$ on the competitive ratio of \pf~\cite{ImKM18} holds for this problem. 
We can also prove that an algorithm 
designed for 
minimizing the total weighted flow time~\cite{ImKMP14} is $32$-competitive for our total weighted completion time objective.
There remains a substantial gap to the lower bound of $2$, and to the best of our knowledge even for the simpler cases of restricted assignment and related machines 
no substantially better constants are known.
In the \emph{clairvoyant} setting, where a job's processing requirement becomes known at the job's release date,
the classical doubling framework by \textcite{HSSW97} yields
the currently best-known polynomial-time algorithm with a competitive ratio of at most~$8$. In a randomized version, this ratio improves to $5.78$~\cite{ChakrabartiPSSSW96}. Although these algorithms compute non-preemptive schedules, it is known that they achieve these competitive ratios even against a preemptive and/or migratory optimal schedule.

In this paper, we drastically improve the analysis of the Proportional Fairness (\pf) algorithm for both, the general polytope scheduling problem (\psp) and several of its important special cases. In particular, we close the
gap between the lower and upper bounds on the competitive ratio for non-clairvoyant scheduling in certain heterogeneous machine environments. Further, we give the first polynomial-time improvements over the nearly 30-year-old bounds on the competitive ratio of the doubling framework for clairvoyant (randomized) scheduling. It is important to note that these improvements are obtained by a deterministic non-clairvoyant algorithm, thereby demonstrating that non-clairvoyance is not a (significant) hurdle.

\subsection{Our Results}

For the general polytope scheduling problem (\psp), we simplify and substantially improve the original analysis of the non-clairvoyant Proportional Fairness algorithm (\pf) by \textcite{ImKM18}. By exploiting the monotonicity of unfinished weight more directly and by optimizing scaling constants, we reduce the bound on the competitive ratio from $128$ down to $27$; the proof is deferred to \cref{app:pf}.

Our main technical contribution is to provide new techniques that lead to further improvements for major special cases of \psp. On a high level, these fall into two categories: \begin{inparaenum}[(i)] \item we show and exploit a monotonicity property of
\pf, and \item we prove and utilize a superadditivity property of the objective function value of an optimal offline solution. \end{inparaenum} A more detailed overview of our techniques is provided in the following subsection; 
here we state and discuss
the results obtained through each of them. \Cref{table:results-online,table:results} summarize our non-clairvoyant scheduling results distinguished by online
and simultaneous job arrival, respectively.

\subsubsection*{Results via Monotonicity}
Our first
results concern \psp\ 
variants for which \pf computes \emph{monotone rates}, a subclass introduced by \textcite{ImKM18}. Informally speaking, in \monpsp the completion of a job does not harm any other job's processing, as it does not cause a decrease of that job's rate.

\begin{definition}[\pf-monotone \psp~\cite{ImKM18}]\label{def:monotone}
  Given job sets $J' \subseteq J$ with corresponding rates $y'$ and $y$ computed by \pf, we say the \psp is \emph{\pf-monotone} (short \monpsp) if for every $j' \in J'$ it holds that $y'_{j'} \geq y_{j'}$.
\end{definition}

We exploit this structural property of \monpsp  rigorously to prove the following theorem, which substantially reduces the previously known bound of $25.74$~\cite{ImKM18}.  

\begin{theorem}\label{thm:main-monotone}
  \pf has a competitive ratio of at most $4$ for minimizing the total weighted completion time for \monpsp.
\end{theorem}

Related machine scheduling as well as multidimensional scheduling with certain utility functions are known to be special cases of \monpsp~\cite{ImKM18}. 
In this paper, we show that the restricted assignment problem is also a \monpsp.
Moreover, we argue that
for restricted assignment as well as related machine scheduling, \pf{}
runs in strongly polynomial time.
We conclude with the following result for these two classical scheduling problems. Most notably, it improves upon the best known competitive ratios for clairvoyant polynomial-time algorithms, which are $8$~\cite{HSSW97,LindermayrMR23} for deterministic algorithms and $5.78$~\cite{ChakrabartiPSSSW96} for randomized algorithms. Remarkably, this is not only the first polynomial-time improvement within nearly thirty years, but it is even obtained by a deterministic non-clairvoyant algorithm. 

\begin{theorem}\label{thm:restricted-related-weighted}
  There is a {strongly} polynomial-time non-clairvoyant online algorithm for minimizing the total weighted completion time on related machines, \abc{Q}{r_j, \pmtn}{\sum w_jC_j}, and for restricted assignment, \abc{R}{r_j,\pmtn, s_{ij}\in \{0,1\}}{\sum w_jC_j}, with a competitive ratio~of~at~most~$4$. 
  \end{theorem}

It seems tempting to hope that \monpsp captures also general unrelated machine scheduling. However, as was suggested in~\cite{IKM15}, this is false; we
give a counterexample in \cref{app:unrelated-not-monotone}.

\begin{table}
  \centering
  \caption{Polynomial-time algorithms for online min-sum completion time scheduling.}\label{table:results-online}
  \begin{threeparttable}
    \begin{tabular}{lr@{~}lr@{~}lr}
      \toprule
      Problem & \multicolumn{2}{r}{old (nclv.)} & \multicolumn{2}{r}{old (clv.)} & this paper (nclv.)  \\
      \midrule
      \psp & 128$^\ddagger$ &\cite{ImKM18} & \gray{128} && 27$^\ddagger$  \\
      \monpsp & 25.74${^*}^\ddagger$ &\cite{ImKM18} & \gray{25.74} && 4$^\ddagger$ \\
      \midrule
      \abc{R}{r_j,\pmtn}{\sum w_j C_j} & 32$^*$ &\cite{ImKMP14} & 8~\cite{HSSW97,LindermayrMR23}, 5.78$^\dagger$ &\cite{ChakrabartiPSSSW96} & 4.62 \\
      \abc{R}{r_j,\pmtn, s_{ij} \in \{0,1\}}{\sum C_j} & \gray{25.74} && \gray{8, 5.78$^\dagger$} && 3 \\
      \midrule
      \abc{Q}{r_j,\pmtn}{\sum w_j C_j} & \gray{25.74} && \gray{8, 5.78$^\dagger$} && 4 \\
      \abc{Q}{r_j,\pmtn}{\sum C_j} & \gray{25.74} && \gray{8, 5.78$^\dagger$} && 3 \\
      \bottomrule
    \end{tabular}
    {\small
      \begin{tablenotes}
          \item[$*$] These bounds are implications of flow time results.
          \item[$\dagger$] Uses randomization. 
          \item[$\ddagger$] 
          The time complexity of \pf depends on the encoding of the polytope.
      \end{tablenotes}
      }
  \end{threeparttable}
 \end{table}

\subsubsection*{Results via Superadditivity}
For our second
group of results, we introduce a subclass of \psp, which is, in contrast to \pf-monotonicity, not defined in terms of \pf.
For a fixed instance,
the defining property makes assumptions on the optimal objective value when all jobs are released simultaneously
as a function of the processing requirements $p=(p_1,\ldots,p_n)$, denoted by $\opt_{0}(p)$.

\begin{definition}[$\alpha$-superadditive \psp]
\label{def:superadditive-psp}
  We say that a \psp is $\alpha$-\emph{superadditive} if for every partition $p = \sum_{\ell=1}^L p^{(\ell)}$ it holds
  that $\sum_{\ell=1}^L \opt_{0}(p^{(\ell)}) \leq \alpha \cdot \opt_{0}(p)$.
\end{definition}

Based on this definition, we introduce a new analysis framework that builds on a decomposition of \pf's schedule into structured instances on which we exploit $\alpha$-superadditivity.
We show the following result for $\alpha$-superadditive \psp.

\begin{theorem}\label{thm:superadditive}
  \pf has a competitive ratio of at most $2\alpha + 1$ for minimizing the total weighted completion time for $\alpha$-superadditive \psp. For uniform release dates, this bound reduces to $2\alpha$.
\end{theorem}

We now outline important implications of this theorem.
First, by showing that unrelated machine scheduling is $1.81$-superadditive, we conclude the following
bound, again improving over the currently best-known bounds of $8$~\cite{HSSW97} for deterministic and $5.78$~\cite{ChakrabartiPSSSW96} for randomized algorithms on polynomial-time clairvoyant online scheduling.

\begin{theorem} \label{theorem:unrelated}
  There is a polynomial-time non-clairvoyant online algorithm for minimizing the total weighted completion time on unrelated machines, \abc{R}{r_j, \pmtn}{\sum w_j C_j}, with a competitive ratio of at~most~$4.62$. For uniform release dates, this bound reduces to $3.62$.
\end{theorem}

Second, we show
that minimizing the total completion time is 1-superadditive
both
on related machines and for restricted assignment.
This implies that \pf achieves the best-possible competitive ratio of $2$ for these problems. To the best of our knowledge, this is the first tight analysis of a non-clairvoyant algorithm on heterogeneous machines for this objective.

\begin{theorem} \label{theorem:2-competitive}
  There is a strongly polynomial-time non-clairvoyant algorithm for minimizing the total completion time with a competitive ratio equal to $2$ for related machines, \abc{Q}{\pmtn}{\sum C_j}, and restricted assignment, \abc{R}{\pmtn, s_{ij} \in \{0,1\}}{\sum C_j}. For non-uniform release dates, the competitive ratio is at most $3$ for both problems.
\end{theorem}

We note that \pf cannot be 2-competitive for non-uniform release dates. In \Cref{app:release-dates-lb} we give an example
showing that its competitive ratio is at least $2.19$,
even on a single machine.

\begin{table}
  \centering
  \caption{Polynomial-time algorithms for non-clairvoyant min-sum completion time scheduling.}\label{table:results}
   \begin{tabularew}{l*{2}{>{\spew{1}{+1}}r@{~}l}r}
     \toprule
     Problem & \multicolumn{2}{r}{old} & \multicolumn{2}{r}{lower bound} & this paper \\
     \midrule
     \abc{R}{\pmtn}{\sum w_j C_j} & 32 &\cite{ImKMP14}& \gray{2} && 3.62 \\
     \abc{R}{\pmtn, s_{ij} \in \{0,1\}}{\sum C_j} & \gray{25.74} && 2 &\cite{MPT94}& 2 \\
     \abc{Q}{\pmtn}{\sum C_j} & \gray{25.74} && 2 &\cite{MPT94} & 2 \\
     \abc{P}{\pmtn}{\sum w_j C_j} & 2 & \cite{BBEM12} & 2 & \cite{MPT94} & 2 \\
     \bottomrule
   \end{tabularew}
\end{table}

\subsubsection*{Implications for Matching Markets} 
As a byproduct of our argumentation, we give new insights on one-sided matching markets with dichotomous utilities. Recently, it has been shown that market equilibria can be computed in strongly polynomial time~\cite{VY21,GTV22b} and, in the case of unit budgets, correspond to optimal solutions of the Nash Social Welfare maximization problem~\cite{GTV22b}. By noting that the market can be formulated as a submodular utility allocation market, we show that both results are also directly implied by the work of \textcite{JainV10}, even for arbitrary budgets.

\subsection{Our Techniques and Intuition}

The Proportional Fairness (\pf) algorithm repeatedly solves a convex program,
and we use its Lagrange multipliers and its KKT conditions to characterize the optimal solutions. Further, we use a \emph{dual fitting} analysis
to compare \pf to an optimal
schedule.
Dual fitting seems to be the most prominent and successful method used in online min-sum scheduling analyses; cf.~\cite{AnandGK12,GuptaKP12,ImKMP14,ImKM18,GGKS19,GuptaMUX20,LindermayrM22,Jaeg23,LindermayrMR23}.
It is straightforward to
incorporate the instantaneous polyhedral constraints of \psp into a standard time-indexed LP relaxation
that describes the scheduling problem with the weighted mean busy time objective
instead of the total weighted completion time~\cite{DyerW90,SchulzS02,ImKM18}. 
To prove a performance guarantee, we 
construct a feasible dual solution 
whose objective value captures a fraction of the algorithm's objective value. 

\subsubsection*{Challenges in the Dual Fitting Approach}
In \psp, the dual of the standard LP has two sets of variables: $(\dualVa_j)_j$ which correspond to the primal constraints that each job has to be completed, and $(\dualVb_{dt})_{d,t}$ which correspond to the primal instantaneous polyhedral constraints (we call them dual packing variables in the following).  
A natural interpretation in the case of a \emph{single machine} can be derived by wiggling the primal constraints~\cite{AnandGK12}: $\dualVa_j$ is proportional to job $j$'s contribution to the total objective value, and the variable $\dualVb_t$ corresponding to the only packing constraint
is proportional to the total weight of available jobs at time $t$.
By setting duals according to this interpretation based on \pf's schedule, one can easily prove that \pf 
is~$4$-competitive for single-machine scheduling. 

Crafting duals for general \psp is substantially more complex, as  pointed out by \textcite{ImKM18}, because, at each time $t$, we must distribute the total weight of available jobs across multiple variables corresponding to the $D$ packing constraints. 
A natural weight distribution
is given by the optimal Lagrange multipliers $(\lagrangevar_{dt})_d$ corresponding to the packing constraints in the Eisenberg-Gale convex program~\cite{eisenberg1959consensus}, which \pf uses to compute the Weighted Nash Social Welfare at time~$t$.
However, this natural dual setup does not work in general because by the dynamics of \pf 
it can happen that a job~$j$ needs more weight distributed to its relevant dual packing variables at some time $t' > t$ to compensate a smaller rate $y_{j}(t') < y_j(t)$; this seems necessary for arguing dual feasibility.
Indeed, there are examples where a job's rate decreases when another job completes in \pf's schedule (cf.~\Cref{app:unrelated-not-monotone}).
The analysis of \textcite{ImKM18}
handles this by carefully redistributing optimal Lagrange multipliers across different times using more complex dual variable assignments; we present a simplified and improved variant 
thereof in \Cref{app:pf}. Our main contribution are the subsequent new methods.

\subsubsection*{Monotonicity}
Our key observation for \monpsp is that these dynamics
have much more structure. In particular, no job requires more weight to be assigned to its relevant dual packing variables at a later time. We show this using the KKT conditions of the Eisenberg-Gale convex program and the fact that a job's rate cannot decrease if another job completes. 
While the latter can happen when a new job $j'$ arrives, $j'$ then already contributes to the objective before its arrival, and we again use monotonicity to show that no job running before $j'$'s arrival suffers from $j'$'s absence.
These observations
admit the natural dual setup outlined above and prove that \pf~is~4-competitive~for~\monpsp~(cf. \Cref{sec:monotone}). 

\subsubsection*{Tight Dual Fitting and Structured Instances}
Breaching the factor of 4 with either of these dual setups seems difficult;
this barrier can also be found in similar dual fitting analyses~\cite{AnandGK12,ImKMP14,ImKM18,GGKS19,GMUX21,LindermayrM22,Jaeg23,LassotaLMS23}.
To the best of our knowledge,
better guarantees for non-clairvoyant algorithms have only been achieved in settings where 
other strong lower bounds are known
\cite{MPT94,DengGBL00,KC03,BBEM12,JagerW24}, which
appears challenging for \monpsp as well as for the general \psp.

If we knew the structural properties of an optimal LP solution, we could potentially overcome this barrier by tailoring the dual packing variables in a way that a job only contributes a fraction of its weight to them that is proportional to its progress in an optimal LP solution. Indeed, we give a tight dual fitting analysis for \pf on a single machine, where we understand the optimal LP solution~\cite{Goe96}.
For \psp however, it seems hard to get sufficient insights on (near-)optimal LP solutions.

We instead pursue a different approach.
Instead of analyzing the whole schedule with one dual fitting argument, we split \pf's objective into several pieces and apply dual fitting to each of them.
Our first key insight to make this work is that a special case of \psp 
admits a tight dual fitting.
This comprises instances in which \pf completes all jobs simultaneously; we call those \emph{structured}. We prove that for structured instances \pf{} yields
an optimal solution to the LP, minimizing the weighted mean busy time, and that this is 2-competitive for the total weighted completion time.

\subsubsection*{Decomposition and Superadditivity}
The second key observation is that we can decompose \pf's objective by splitting the \pf schedule at release dates and completion times into structured instances. 
This does not work for any schedule and uses specific properties: \begin{inparaenum}[(i)] \item the \pf schedule after time $t$ corresponds to the \pf schedule for the instance where the processing requirements of all jobs are reduced by the processed amount before time $t$ in the original \pf schedule,
and \item \pf
assigns an available job a \emph{positive} rate, hence
its flow time can be 
expressed as the sum of its completion times in the structured subinstances in which it appears.
\end{inparaenum}
Finally, we bound the overall optimal objective value by the optimal objective values of the subinstances via superadditivity.

The primary challenge in showing superadditivity for machine scheduling problems
lies in handling migrations.
At first sight one may think that allowing job migrations makes a problem easier, 
since an algorithm has more freedom to revert bad decisions, especially in online settings. We even show that any non-migratory non-clairvoyant algorithm has an unbounded competitive ratio (cf.\ Theorem~\ref{thm:migration-necessary}).
However, at the same time migrations allow 
to create more complex optimal schedules, which makes analyses way more difficult.
Indeed, there are many indications in literature that scheduling unrelated machines with migrations
is much harder than without:
\begin{inparaenum}[(i)] \item if migration is allowed, the unweighted offline problem is APX-hard~\cite{Sitters17}, while otherwise it is polynomial-time solvable~\cite{BrunoCS74,Horn73},
\item strong time-indexed LP relaxations no longer lower bound the optimal objective value if migration is allowed~\cite{SchulzS02}, and
\item best-known offline approximation ratios for the weighted problem are stronger if migration is disallowed~\cite{li2024approximating,Harris24,ImL16,Sitters17,Sku01,SchulzS02,ImL23}. \end{inparaenum}

Our key trick for unrelated machines is to
evade migrations by showing 1-superadditivity for non-preemptive schedules, which are more structured and, thus, easier to handle. 
To go back to the optimal migratory objective value, we use the concept of the \emph{power of preemption}, that is, the factor between an optimal non-preemptive objective value and an optimal migratory objective value. For unrelated machines, the currently best-known upper bound by \textcite{Sitters17} on this ratio is~1.81.

The analysis 
via non-preemptive schedules
always loses the power of preemption as a
factor, hence rules out tight competitiveness bounds of $2$ whenever this factor is larger than $1$ because this ratio can already be tight when comparing to the optimal non-preemptive schedule.
To still achieve a competitive ratio of $2$ for related machines, where the power of preemption is at least 1.39~\cite{EpsteinLSS17}, we instead directly analyze the optimal migratory solution. This is possible because on related machines an optimal solution
has many nice structural properties~\cite{gonzalez1977optimal}.

\subsection{Further Related Work}

A technique related to our analysis framework has been introduced by
\textcite{DengGBL00}, and was also used later~\cite{BBEM12,JagerW24}. These works analyze non-clairvoyant algorithms with an inductive argument on the number of jobs by showing that truncating the schedule at the first completion time maintains the bound on the competitive ratio. Compared to our problem, these works can use simpler combinatorial lower bounds on an optimal solution to make this approach work. Moreover, they assume that all jobs are available in the beginning, and it is not clear how to integrate non-uniform release dates into the inductive argument.

Non-clairvoyant scheduling to minimize the total weighted completion time on parallel identical machines is well-understood~\cite{MPT94,DengGBL00,BBEM12,MoseleyV22}, and even in the setting with precedence constraints the gap to the lower bound of 2 was recently narrowed~\cite{GGKS19,LassotaLMS23,JagerW24}.
While we show that even on a single machine, \pf cannot be 2-competitive in the presence of release dates, we note that
this bound is 
achieved by some other single-machine scheduling algorithm~\cite{JagerSSW24}.
For heterogeneous machines, besides the analysis 
by \textcite{ImKM18} for \pf (and a slight improvement~\cite{LindermayrMR23}), there are several analyses for the more general total flow time objective with speed augmentation~\cite{ImKMP14,Gupta2010,GuptaIKMP12,ImKM18}, 
whose bounds can be translated to the completion time objective.
Furthermore, \textcite{GKS21} presented a non-clairvoyant algorithm for related machine scheduling with a special type of precedence constraints.

For scheduling on unrelated machines in the clairvoyant online setting,
\textcite{BienkowskiKL21} improved in a very careful analysis the aforementioned doubling framework~\cite{HSSW97} the bound on the competitive ratio (when ignoring the requirement on polynomial-time algorithms) from $4$~\cite{HSSW97} to $3$ for deterministic algorithms and from $2.886$~\cite{ChakrabartiPSSSW96} to $2.443$ for randomized algorithms.
Furthermore, it is known that a natural polynomial-time extension of the preemptive WSPT rule to unrelated machines is $8$-competitive~\cite{LindermayrMR23}.
Other results are several greedy algorithms which assign jobs at their arrival immediately to machines and never migrate them, and, thus, were all analyzed against a non-migratory optimal solution, for which stronger LP relaxations are possible~\cite{AnandGK12,GMUX21,Jaeg21,LindermayrM22}. Moreover, there is a $2$-competitive algorithm for a special case of related machines for which the ordered speed vector $s_1 \ge \cdots \ge s_m$ does not decrease too quickly~\cite{LCXZ09}.
\textcite{ChenImPetety24gd} presented very recently a clairvoyant online algorithm for several instances of PSP for minimizing the total weighted flow time with speed augmentation.

\section{Preliminaries} \label{sec:preliminaries}

\subsection{PSP and Unrelated Machine Scheduling}
In the polytope scheduling problem, jobs $j=1,\dotsc,n$ arrive online at their release dates~$r_j \ge 0$. When a job is released, its weight~$w_j > 0$ as well as the column~$B_j$ from the matrix $B = (b_{dj})_{d,j} \in \mathbb Q_{\ge 0}^{D \times n}$ become known. 
A schedule assigns at every time~$t \ge 0$ a processing rate~$y_j(t) \ge 0$ to each job~$j$ with $r_j \le t$, while all other jobs~$j$ have $y_j(t) = 0$. The resulting completion time~$C_j$ of a job~$j$ is the first time~$t'$ such that $\int_{r_j}^{t'} y_j(t)\,\mathrm d t \ge p_j$, where $p_j$ denotes the processing requirement of job $j$. The processing rate vector $y(t)$ must satisfy, at any time~$t$, the constraints~$B \cdot y(t) \le \bone$, $y(t) \ge \mathbf 0$ defining the
polytope~$\cP$. The objective is to minimize the sum of weighted completion times~$\sum_j w_j C_j$ of all jobs.

In the special case of unrelated machine scheduling, we are given $m$ machines where every job~$j$ runs on every machine $i$ at speed $s_{ij} \geq 0$. At any time a job can be scheduled on at most one machine, and on every machine can run at most one job. Thus, we can model any schedule for a fixed time unit $t$ using variables $x_{ijt} \geq 0$ that indicate whether job $j$ runs on machine $i$ at time~$t$ and constraints $\sum_{j=1}^n x_{ijt} \le 1$ for every machine $i$ and $\sum_{i=1}^m x_{ijt} \le 1$ for every job $j$.
The rate of job $j$ at time~$t$ is thus equal to $y_j(t) = \sum_{i=1}^m s_{ij} x_{ijt}$, and we can write the polytope of feasible  rate-allocation pairs as
\[ 
\mathcal Q \coloneqq \biggl\{(y,x) \in \RR_{\ge 0}^{n} \times \RR_{\ge 0}^{m \times n} \biggm| y_j = \sum_{i=1}^m s_{ij} x_{ij}\,\forall j \in [n],\ \sum_{j=1}^n x_{ij} \le 1\, \forall i \in [m],\ \sum_{i=1}^m x_{ij} \le 1\, \forall j \in [n]\biggr\} \ .
\]
Since $\mathcal Q$ is downward-closed and full-dimensional in $y$,
we can obtain a matrix $B \in \Q_{\geq 0}^{D \times n}$ such that the projection of  $\mathcal Q$ to $y$ is equal to $\{y \in \RR_{\geq 0}^n \bigm| By \leq \bone \}$ by applying Fourier-Motzkin elimination to $\mathcal Q$, showing that unrelated machine
scheduling is indeed a special case of \psp.

Any fractional assignment~$x$ can be efficiently decomposed into integral assignments, which we can feasibly be scheduled with preemptions and migrations during a discrete time interval around~$t$. We emphasize that migration is essential for non-clairvoyant
algorithms
to achieve a constant competitive ratio for unrelated machine  scheduling. 
We show that this is true even in the special case of related machines. The proof is given in \cref{app:lb-migratory}.

\begin{theorem}\label{thm:migration-necessary} 
  Any non-migratory non-clairvoyant algorithm for minimizing the total completion time of $n$ jobs on related machines, \abc{Q}{\pmtn}{\sum C_j}, has a competitive ratio of at least $\Omega(\sqrt{n})$.
\end{theorem}

\subsection{Proportional Fairness}
The Proportional Fairness strategy (\pf), which has been analyzed for \psp by \textcite{ImKM18}, computes at every release time and completion time the processing rates~$y_j(t)$ of all jobs~$j$ that have been released but not yet completed. These jobs are constantly processed at the computed rates until a new job is released or completed. We denote by $J(t) \coloneqq \{j \in [n] \mid r_j \leq t < C_j\}$ the set of \emph{available jobs} at time~$t$ in \pf's schedule. 
Although the completion times of the jobs~$j$ with~$C_j > t$ are still unknown at
time~$t$, it is already known at this point whether a job~$j$
belongs to~$J(t)$, so that this
set can be used by \pf. 

As mentioned in the introduction, at any instant~$t$,
\pf maximizes the Weighted Nash Social Welfare, which is the weighted geometric mean of the 
allocated job rates among all feasible rate vectors~$y(t) \in \mathcal P$ supported on $J(t)$.
It is the solution to the following convex program applied to~$J \coloneqq J(t)$:
{
 \renewcommand{\arraystretch}{1.5}
 \[\begin{array}{rr>{\displaystyle}rcl>{\quad}l}
  (\CP(J))& \operatorname{maximize} &\multicolumn{3}{l}{\displaystyle\sum_{j \in J} w_j \cdot \log(y_j)} \\
  &\text{subject to} &\sum_{j \in J} b_{dj} \cdot y_j &\le& 1 &\forall d \in [D] \\
  &&y_j &\ge& 0 &\forall j \in J
 \end{array}\]
}%
Note that the \pf strategy
may not be executed exactly algorithmically 
because for some rational inputs the convex program may have only irrational solutions~\cite[Example~48]{JainV10}. While this is no issue for many settings, in general $(\CP(J))$ can be solved to arbitrary precision~$\delta > 0$ in time polynomial in the encoding length of $B$ and $w$ and polynomial in $\log \frac 1 \delta$ using the ellipsoid method~\cite{grotschel2012geometric,vishnoi2021algorithms}.
As a consequence, we can algorithmically implement \pf so that for every $\varepsilon > 0$ we can compute rates in polynomial time and only lose a factor of $1+\varepsilon$ in the competitive ratio for \psp. This is possible by
 approximately computing rates for a slower machine and increase them by $\delta$ afterwards.

Since \pf computes rates at most $2n$ times, 
we say that it runs in polynomial time. Note that an actual schedule for these rates requires a pseudo-polynomial number of preemptions. However, with an implementation in geometrically increasing time windows one can reduce the number to a polynomial at the cost of an arbitrarily small constant in the competitive ratio, cf.~\cite{MPT94}.

We can assume w.l.o.g.\ that $p_j > 0$, because every reasonable algorithm can complete a job $j$ with $p_j=0$ at time 0. Similarly, if there is a job $j$ such that $y_j = 0$ for all $y \in \cP$, no solution can complete job $j$. Thus, we assume that this is not the case. Then the definition of $(\CP)$ implies
the following observation.
\begin{observation}\label{obs:positive-rates}
Let $y$ be the solution to $(\CP(J))$. Then, $y_j > 0$ for all $j \in J$.
\end{observation}

As the analysis of \textcite{ImKM18}, our analysis uses the optimal Lagrange multipliers~$(\lagrangevar_{d})_d$ of the packing constraints, which satisfy with the optimal solution~$y$ the following KKT conditions~\cite{BV2014}:
\begin{align}
 \frac{w_j}{y_j} - \sum_{d=1}^D b_{dj} \lagrangevar_{d} &= 0 \quad \forall j \in J \label{psp-kkt1}\\
 \lagrangevar_{d} \cdot \biggl(1- \sum_{j \in J} b_{dj} y_{j} \biggr) &= 0 \quad \forall d \in [D]  \label{psp-kkt2} \\
 \lagrangevar_{d} & \geq 0 \quad \forall d \in [D] \label{psp-kkt3}
\end{align}
By Observation~\ref{obs:positive-rates}, the optimal solution of $(\CP(J))$ assigns every job a positive rate, hence we can omit
the Lagrangian multipliers of the non-negativity constraints of $(\CP(J))$ in the KKT conditions.

\begin{lemma}[{cp.~\cite[Lemma~3.3]{ImKM18}}]\label{lemma:kkt-multiplier-bound} 
 Let $J$ be a set of jobs, and let $\lagrangevar$ be optimal Lagrange multipliers for $(\CP(J))$. Then $\sum_{d=1}^D \lagrangevar_{d} = \sum_{j \in J} w_j$.
\end{lemma}
  
\begin{proof}
 Let $y$ be an optimal solution to $(\CP(J))$. Then
 \begin{align*}
   \sum_{d=1}^D \lagrangevar_{d}
   = \sum_{d=1}^D \lagrangevar_{d} \sum_{j\in J} b_{dj} y_{j}
   =  \sum_{j\in J}  y_{j} \sum_{d=1}^D b_{dj} \lagrangevar_{d}
   = \sum_{j\in J}  y_{j} \frac{w_j}{y_{j}}
   = \sum_{j \in J} w_j
 \end{align*}
 by using the KKT condition~\eqref{psp-kkt2} in the first equality and~\eqref{psp-kkt1} in the third equality.
\end{proof}

\subsection{LP relaxations and Dual Fitting}
In the dual fitting analysis,
\pf is compared to the optimal solution that runs at speed~$\frac{1}{\kappa}$ for some parameter $\kappa \geq 1$. We assume that the time is scaled in a way that all release dates and completion times in this schedule and in \pf's schedule are integral (which is possible assuming rational input). In this case, the following linear program is a relaxation of the speed-scaled \psp, where the variables~$y_{jt}$ indicate the total amount of processing that job~$j$ receives in the interval~$[t,t+1]$. It uses the total weighted mean busy time as
an underestimation of the total weighted completion time~\cite{DyerW90,Goe96}. 
{
 \renewcommand{\arraystretch}{1.5}
 \[\begin{array}{rr>{\displaystyle}rc>{\displaystyle}l>{\quad}l}
  (\LP(\kappa))& \operatorname{minimize} &\multicolumn{3}{l}{\displaystyle\sum_{j \in J} w_{j} \sum_{t \geq r_j} \frac{y_{jt}}{p_j}  \Bigl(t + \frac{1}{2} \Bigr)} \\
  &\text{subject to} &\sum_{t \geq r_j} y_{jt} &\ge& p_j & \forall j \in J  \\
  && \sum_{j \in J} b_{dj} \cdot y_{jt} &\le& \frac{1}{\kappa} &\forall d \in [D],\ \forall t \geq 0 \\
  && y_{jt}  &\ge& 0 &\forall  j \in J,\ \forall t \geq r_j 
 \end{array}\]%
}%
The scalability of completion times immediately implies that the optimal objective value of $(\LP(\kappa))$ lower bounds $\kappa \cdot \opt$ for every $\kappa \geq 1$.
The dual of $(\LP(\kappa))$ can be written as follows.
{
 \renewcommand{\arraystretch}{1.9}
 \[\begin{array}{rr>{\displaystyle}rc>{\displaystyle}l>{\quad}l}
  (\DLP(\kappa))& \operatorname{maximize} &\multicolumn{3}{l}{\displaystyle\sum_{j \in J} \dualVa_j - \sum_{d = 1}^D \sum_{t \geq 0} \dualVb_{dt}} \\
  &\text{subject to} &\frac{\dualVa_j}{p_j} - \frac{w_j}{p_j} \Bigl(t + \frac{1}{2} \Bigr) &\le& \kappa \sum_{d=1}^D b_{dj}  \dualVb_{dt} &\forall j \in J,\ \forall t \geq r_j \\
  &&\dualVa_j, \dualVb_{dt} &\ge& 0 &\forall j \in J,\ \forall t \geq 0,\ \forall d \in [D]
 \end{array}\]
}%

For the analysis of \pf we introduce for every $t \ge 0$ the notation $U(t) \coloneqq \{j \in J \mid C_j > t\}$ for the set of \emph{unfinished jobs} at time~$t$ and $W(t) \coloneqq \sum_{j \in U(t)} w_j$ for their total weight. 
Since we assumed that \pf's completion times are integral, it holds that $\sum_{t \geq 0} W(t) = \sum_{j \in J} w_j C_j$.

\subsection{Offline Complexity}
The offline \psp is
APX-hard, even for uniform release dates.
This follows for example from the hardness of preemptive unrelated machine scheduling~\cite{Sitters17}.
On the positive side, constant-factor approximation algorithms can be obtained via standard techniques~\cite{ImKM18}. For the sake of completeness, we demonstrate in \Cref{app:offline-psp} how randomized $\alpha$-point rounding~\cite{SchulzS97,QueyranneS02} applied to an optimal solution of a variant of $(\LP(1))$ yields a $(2+\varepsilon)$-approximation for \psp. The proof also implies that $(\LP(1))$ lower bounds the optimal objective value within a factor of~$1/2$.

\section{\pf-Monotone \psp} \label{sec:monotone}

In this section, we consider the class of \pf-monotone \psp~(\monpsp, cf.~Definition~\ref{def:monotone}) and prove one of our main results, Theorem~\ref{thm:main-monotone}, stating that \pf has a competitive ratio of at most $4$ for this class. Further, we show implications of this result for specific machine scheduling problems in \monpsp, namely, related machine scheduling and restricted assignment (Theorem~\ref{thm:restricted-related-weighted}). We also point out the implications of our work for matching markets in this \lcnamecref{sec:monotone} (\cref{cor:hz-convex-program,cor:hz-strongly-polynomial}).

Recall that \monpsp consists of instances to \psp whose polytope~$\cP$ has the property that for all possible sets~$J' \subseteq J$ of available jobs with corresponding rate vectors~$y'$ and $y$ computed by \pf it holds that $y'_j \ge y_j$ for all $j \in J'$ (cf.~Definition~\ref{def:monotone}).

\subsection{Competitive Analysis}\label{subsec:monotone-competitive}

To show the upper bound of $4$ on the competitive ratio (Theorem~\ref{thm:main-monotone}), we fix an arbitrary instance of \monpsp and the corresponding \pf schedule. Let $C_j$ denote the completion time of job $j$ in that schedule, and let $\alg:=\sum_{j=1}^n w_jC_j$ denote the
objective function value of \pf.
For each time~$t$, in addition to the actual rate vector~$y(t)$ of \pf, we consider the rate vector~$\hy(t) \in \cP$ that \pf would choose if all jobs~$j \in U(t)$ were available at time $t$.
In other words, $y(t)$ is the optimal solution to $(\CP(J(t)))$ and $\hy(t)$ is the optimal solution to $(\CP(U(t)))$. Note that, while we use these rates for the analysis, during the actual execution, \pf cannot use them because it has no access to unreleased jobs. Moreover, the set $U(t)$ always refers to the unfinished jobs in the actual PF schedule and not in the schedule with rates~$\hy(t)$. Let $\hlagrangevar_{d}(t)$, $d \in [D]$, be the optimal Lagrange multipliers corresponding to $(\CP(U(t)))$.

Using $J(t) \subseteq U(t)$ for any time $t$ and monotonicity, we make the following observation.

\begin{observation}\label{lemma:monotone-rates}
  At every time $t$ and for every $j \in J(t)$ it holds that $\hy_{j}(t) \leq y_{j}(t)$.
\end{observation}

We now perform a dual fitting argument via $(\DLP(\kappa))$ for arbitrary $\kappa \geq 1$, which we later set to $2$.
To this end, we consider the following assignment of dual variables:
\begin{itemize}
 \item $\dualSa_j \coloneqq w_j C_j$ for every job $j \in J$,
 \item $\dualSb_{dt} \coloneqq \frac{1}{\kappa} \hlagrangevar_{d}(t)$ for every $d \in [D]$ and time $t \geq 0$.
\end{itemize}

\begin{lemma}\label{lemma:monoton-psp-dual-objective}  
It holds that $\bigl(1 - \frac{1}{\kappa}\bigr) \alg = \sum_{j \in J} \dualSa_j - \sum_{d=1}^D \sum_{t \geq 0} \dualSb_{dt}$.
\end{lemma}

\begin{proof}
 First, note that $\sum_{j \in J} \dualSa_j = \alg$. Moreover, the definition of $\dualSb_{dt}$ and Lemma~\ref{lemma:kkt-multiplier-bound} imply that $\sum_{d=1}^D \dualSb_{dt} = \frac{1}{\kappa} W(t)$ at every time $t$. Thus, $\sum_{t \geq 0} \sum_{d=1}^D \dualSb_{dt} = \frac{1}{\kappa} \alg$, which concludes the proof.
\end{proof}

\begin{lemma}\label{lemma:monoton-psp-dual-feasibility}
 The solution $(\dualSa_j)_{j}$ and $(\dualSb_{dt})_{d,t}$ is feasible for $(\DLP(\kappa))$.
\end{lemma}

\begin{proof}
 The dual variables as defined above are clearly non-negative.
 We now verify the dual constraint.
   Fix a job $j$ and a time $t \geq r_j$. Then,
   \begin{align*}
     \frac{ \dualSa_j }{p_j} - \Bigl(t + \frac 1 2\Bigr) \cdot \frac{w_j}{p_j}
     \leq (C_j-t) \cdot \frac{w_j}{p_j}
     \le \sum_{t' = t}^{C_j-1} \frac{w_j}{p_j} 
     =  \sum_{t' = t}^{C_j-1} \frac{\hy_{j}(t')}{p_j} \cdot \frac{w_j}{\hy_{j}(t')} 
     \leq \sum_{t' = t}^{C_j-1}  \frac{y_{j}(t')}{p_j}   \cdot \frac{w_j}{\hy_{j}(t')}\ .     
 \end{align*}
 The last inequality uses Lemma~\ref{lemma:monotone-rates}.
 Now observe that for every time $t' \geq t$ it holds that $U(t') \subseteq U(t)$, which implies $\hy_{j}(t') \geq \hy_{j}(t)$ via monotonicity.
 Thus, the above is at most 
\begin{align*}
  \frac{w_j}{\hy_{j}(t)}  \sum_{t' = t}^{C_j-1}  \frac{y_{j}(t')}{p_j} 
  =    \bigg(  \sum_{d=1}^D b_{dj} \hlagrangevar_{d}(t) \bigg) \sum_{t' = t}^{C_j-1} \frac{y_{j}(t')}{p_j} 
  \leq  \kappa  \sum_{d=1}^D b_{dj} \dualSb_{dt} \ ,
\end{align*}
 where the equality uses the KKT condition~\eqref{psp-kkt1} for $(\CP(U(t)))$, and the inequality holds because $j$ receives at most as much processing rate as it requires.
 This concludes the proof.
\end{proof}
  
\begin{proof}[Proof of Theorem~\ref{thm:main-monotone}]
	  Weak duality and Lemma~\ref{lemma:monoton-psp-dual-objective} and~\ref{lemma:monoton-psp-dual-feasibility} imply
  \begin{align*}
      \kappa \cdot \opt 
      \geq  \sum_{j \in J} \dualSa_j - \sum_{d=1}^D \sum_{t \geq 0} \dualSb_{dt} \geq \Bigl(1 -  \frac{1}{\kappa}\Bigr) \cdot \alg \ .
  \end{align*}
  Setting $\kappa = 2$  implies $\alg \leq 4 \cdot \opt$.
\end{proof}

\subsection{Implications for Machine Scheduling and Matching Markets} \label{subsec:monotone-applications}

Now we prove Theorem~\ref{thm:restricted-related-weighted} and thereby present new insights for one-sided matching markets.
We first note that both scheduling problems, on related machines (\abc{Q}{r_j, \pmtn}{\sum w_jC_j}) and with restricted assignment (\abc{R}{r_j, \pmtn, s_{ij} \in \{0,1\}}{\sum w_j C_j}), fall under \monpsp. Hence, \pf is $4$-competitive by Theorem~\ref{thm:main-monotone}. For scheduling on related machines, this has been shown by \textcite{ImKM18} via a connection to abstract markets, which we present in the following. Then, we prove that scheduling with restricted assignment falls also under \monpsp. Finally, we argue that \pf can be implemented in strongly polynomial time for both problems.

\subsubsection*{Monotonicity via Submodular Utility Allocation Markets}
In \emph{uniform utility allocation}~(UUA) \emph{markets}, introduced by \textcite{JainV10}, $n$ buyers~$j$ with budgets~$w_j$ want to maximize their own utility~$y_j$, but there are constraints of the form $\sum_{j \in S} y_j \le v(S) \ \forall S \subseteq [n]$ for some set function $v \colon 2^{[n]} \to \RR_{\ge 0}$, coupling the utilities of different buyers. Every constraint can post a price~$\eta_S$ that each buyer~$j \in S$ has to pay per unit of utility. That means, the total amount to be paid by $j$ is the sum of prices of the constraints affecting $j$. A \emph{market equilibrium} consists of a feasible utility allocation~$y$ together with prices~$\eta \ge \mathbf 0$ such that only tight constraints have a positive price and each buyer exactly uses up their budget, thus maximizing their own utility.

\Textcite{JainV10} showed for UUA markets that optimal solutions to the Weighted Nash Social Welfare  problem correspond to market equilibria. This problem is exactly the problem~$(\CP([n]))$ for the polytope $\cP = \{y \in \RR^n_{\ge 0} \mid \sum_{j \in S} y_j \le v(S)\ \forall S \subseteq [n]\}$. Moreover, they proved that if $v$ is monotone, submodular and $v(\emptyset) = 0$, i.e., if $\mathcal P$ is a polymatroid, then the market (termed \emph{submodular utility allocation} (SUA) \emph{market}) is competition-monotone, meaning that if some budgets are reduced,
the utilities of other buyers cannot decrease. By taking buyers as jobs and utilities as processing rates, this implies that PSP is \pf-monotone on $\mathcal P$ because we can model the restriction to a subset of jobs by setting the weight of the other jobs to $0$. 
Consequently, 
it suffices to show that the polytopes associated to our scheduling problems are polymatroids.
\textcite{ImKM18} proved this for scheduling on related machines.
We now show this statement for scheduling with restricted assignment.
This implies that \abc{R}{r_j, \pmtn, s_{ij} \in \{0,1\}}{\sum w_j C_j} is PF-monotone.

\begin{lemma}\label{lemma:restricted-polymatroid}
 For any instance of preemptive scheduling with restricted assignment, the associated polytope~$\cP$ is a polymatroid.
\end{lemma}
\begin{proof}
 Consider the bipartite graph~$G \coloneqq (J \mathbin{\dot\cup} [m], E)$ with an edge $\{j,i\} \in E$ whenever job~$j$ can be processed on machine~$i$. Then, letting $\delta$ denote the set of incident edges, we have \[\cP = \biggl\{y \in \RR_{\ge 0}^{n} \biggm| \exists x \in \RR_{\ge 0}^{E}: y_j = \sum_{e \in \delta(j)} x_{e}\,\forall j \in J ,\ \sum_{e \in \delta(i)} x_{e} \le 1\, \forall i \in [m],\ \sum_{e \in \delta(j)} x_{e} \le 1\, \forall j \in J \biggr\}\ ,\]
 i.e., $\cP$ contains for every fractional matching of $G$ the vector~$y$ of fractional covering rates of all nodes~$j$ on the left-hand side~$J$ of $G$. By the Birkhoff-von Neumann Theorem, every fractional matching is a convex combination of integral matchings. Then the covering rates of the nodes are also the corresponding convex combinations of the incidence vectors of the subsets of nodes from~$J$ covered by the integral matchings. Therefore, $\mathcal P$ is the convex hull of the incidence vectors of subsets of~$J$ covered by a matching in $G$. This is exactly the independence polytope associated with the transversal matroid of $G$ (cf.~\cite{Oxl11}), which is in particular a polymatroid.
\end{proof}

For matching markets with dichotomous utilities, which correspond to scheduling with restricted assignment, \textcite{GTV22b} gave a direct proof that (in the unit budget case) optimal solutions to the Nash Social Welfare problem (formulated in terms of the allocation $x$) correspond to market equilibria, known as Hylland-Zeckhauser (HZ) equilibria~\cite{HZ79}. 
Lemma~\ref{lemma:restricted-polymatroid} implies that such a market can be modeled as a UUA market, and, thus,
even in the more general case where buyers have different budgets, the work of \textcite{JainV10} also implies this equivalence between the HZ equilibria and the Weighted Nash Social Welfare solution. Hence, we obtain the following slightly stronger result using an arguably simpler argumentation.

\begin{corollary} \label{cor:hz-convex-program}
For one-sided matching markets with dichotomous utilities and arbitrary budgets, an HZ equilibrium is an optimal solution to $(\CP)$, and every optimal solution to $(\CP)$ can be extended to an HZ equilibrium.
\end{corollary}

\subsubsection*{Strongly Polynomial Algorithms}
Since the rate polytopes for both scheduling on related machines and with restricted assignment are polymatroids, we can use the combinatorial algorithm of \textcite{JainV10} for SUA markets, which is entirely stated in terms of utilities and based on submodular function minimization,
to compute a market equilibrium, and, thus, a rate allocation~$y$ that is optimal for $(\CP)$.
A corresponding machine allocation~$x$ can be computed easily by finding a feasible fractional assignment. 
This completes the proof of Theorem~\ref{thm:restricted-related-weighted}.

This algorithm is thus an alternative to the recently proposed strongly polynomial-time algorithm for computing a market equilibrium in matching markets with dichotomous utilities proposed by \textcite{VY21}, which directly incorporates the allocation~$x$. Since this problem exactly corresponds to \pf's allocation problem for scheduling with restricted assignment, an extension of their algorithm to arbitrary budgets~\cite{GTV22b} can also be used in \pf.
While their algorithm is conceptually simpler than the algorithm by \textcite{JainV10} for general SUA markets, the latter implies the following result via an arguably simpler argumentation.

\begin{corollary} \label{cor:hz-strongly-polynomial}
  An HZ equilibrium in one-sided matching markets with dichotomous utilities can be computed in strongly polynomial time.
\end{corollary}

For scheduling on related machines, we demonstrate in \cref{app:pf-combinatorial}, analogously to the algorithms for scheduling with restricted assignment resp.\ one-sided matching markets with dichotomous utilities~\cite{VY21,GTV22b}, that also for this problem there is a simpler and faster algorithm that directly works with the allocation~$x$.
Moreover, our combinatorial algorithm
reveals that the known $2$-competitive non-clairvoyant algorithm for identical parallel machines~\cite{BBEM12} is a special case of \pf.

\section{Superadditive \psp} \label{sec:superadditive}

In this section, we consider an $\alpha$-superadditive \psp instance for some $\alpha \ge 1$. That means that for an arbitrary partition of the processing requirements, the sum of the optimal objective values of the subinstances is at most $\alpha$ times the optimal objective value of the whole instance; cf.~Definition~\ref{def:superadditive-psp}. We first present the analysis framework and prove that \pf has a competitive ratio of at most $2\alpha$ for $\alpha$-superadditive \psp (Theorem~\ref{thm:superadditive}). Subsequently, we apply this results and analysis framework to several machine scheduling problems by proving bounds on the superadditivity.

\subsection{Framework}\label{sec:framework}

The proof of Theorem~\ref{thm:superadditive} consists of three steps: first we decompose the instance into structured subinstances
according to the \pf schedule. 
Second we show competitiveness bounds for the resulting structured subinstances, and third we combine the obtained bounds using the $\alpha$-superadditivity.

\subsubsection{Decomposing the \pf Schedule} \label{subsubsec:splitting}
Let $C_j$, $j \in [n]$, denote the completion times of the jobs in the \pf schedule, and let $E_1 < \ldots < E_{L+1}$ be the times when jobs arrive or complete in the \pf schedule.
We assume w.l.o.g.\ that~$E_1 = 0$. For every $1 \leq \ell \leq L$ and job $j$ we denote by $p_j^{(\ell)}$ the amount of processing which $j$ receives during~$[E_{\ell},E_{\ell+1}]$. Let $\alg(p)$ be the objective value of the schedule constructed by \pf,
and let $\alg_{0}(p^{(\ell)})$ be the objective value of the PF schedule applied to jobs with processing requirements~$p^{(\ell)} = (p_1^{(\ell)},\ldots,p_n^{(\ell)})$ available at time~$0$. Note that $p = \sum_{\ell=1}^{L} p^{(\ell)}$, because \pf computes a feasible schedule. We decompose \pf's objective value according to the following lemma.

\begin{lemma} \label{lemma:splitting}
 It holds that 
 \[\alg(p) = \sum_{\ell=1}^{L} \biggl(\alg_{0}(p^{(\ell)}) + (E_{\ell+1} - E_\ell) \mathrlap{{}\cdot{}} \sum_{j \in U(E_\ell) \setminus J(E_\ell)} w_j\biggr)\ .\]
\end{lemma}
\begin{proof}
 Observation~\ref{obs:positive-rates} implies that a job $j$ is available during an interval~$(E_\ell, E_{\ell+1}]$ if and only if $p_j^{(\ell)} > 0$.
 Hence, the jobs available after time~$0$ in the subinstance  with processing requirements $p^{(\ell)}$ and uniform release dates, denoted by $I^{(\ell)}$, are exactly the jobs available in the original instance during $(E_\ell, E_{\ell+1}]$. 
 Since \pf takes an instantaneous resource allocation view, the composed processing rates
 depend only on the weights of the currently available jobs, i.e., the rate which job $j$ receives in the original instance during $(E_\ell,E_{\ell+1}]$ is equal to the rate which $j$ receives in $I^{(\ell)}$. Moreover, the rates are constant during the considered interval, so that for every available job $j$ it only happens at the end of the interval that its total processing rate received within the interval reaches $p_j^{(\ell)}$. For $I^{(\ell)}$ this means that all jobs with a positive processing requirement finish exactly at time~$E_{\ell+1} - E_\ell$ and all jobs with a processing requirement equal to zero finish at time $0$. Thus, $\alg_{0}(p^{(\ell)}) = \sum_{j \in J(E_\ell)} w_j \cdot (E_{\ell+1} - E_\ell)$. Therefore, we conclude with
 \begin{align*}
  \alg(p) &= \sum_{j=1}^n w_j C_j = \sum_{j=1}^n w_j \sum_{\ell: E_\ell < C_j} (E_{\ell+1} - E_\ell) = \sum_{\ell=1}^{L} (E_{\ell+1} - E_\ell) \sum_{j \in U(E_\ell)} w_j \\
  &= \sum_{\ell=1}^{L} \biggl(\alg_{0}(p^{(\ell)}) + (E_{\ell+1} - E_\ell) \cdot \sum_{j \in U(E_\ell) \setminus J(E_\ell)} w_j\biggr)\ .
 \end{align*}
\end{proof}

\subsubsection{Analyzing Structured Instances}

We now show that \pf is 2-competitive against $(\DLP(1))$ for \psp instances where all jobs are available at time $0$ and \pf completes all jobs with positive processing requirements at the same time. That is, the rate~$y_j(0)$ allocated to each job~$j$ does not change and is proportional to its processing requirement~$p_j$. Any job with a processing requirement equal to $0$ completes at time $0$ and, thus, can w.l.o.g.\ be removed from the instance.

We fix a \psp instance of jobs~$J$ with these properties and assume that in the optimal schedule all completion times are integral. Let $C$ be the common completion time in the \pf schedule and also assume by scaling that $C$ is an integer. Since the rates and Lagrange multipliers do not change, we write $y_j = y_{j}(t)$, $j \in J$, and $W = W(t)$ and $\lagrangevar_d$ for the Lagrange multipliers corresponding to $(\CP(J(t)))$ for all $t \in [0,C)$.
We perform a dual fitting argument using the following assignment:
  \begin{itemize}
    \item $\dualSa_j \coloneqq w_j C$ for every job $j \in J$,
    \item $\dualSb_{dt} \coloneqq \bigl(1 - \frac{1}{C}(t+\frac{1}{2})\bigr) \cdot \lagrangevar_{d}$ for every $d \in [D]$ and time $t \in \{0,\ldots,C-1\}$.
  \end{itemize}

\begin{lemma}\label{lemma:common-c-dual-objective}
    It holds that $\frac{1}{2} \cdot \alg = \sum_{j \in J} \dualSa_j - \sum_{t = 0}^{C-1} \sum_{d=1}^D \dualSb_{dt} $.
\end{lemma}

\begin{proof}
  Clearly, $\sum_{j \in J} \dualSa_j = \alg$. Further, at every time $t$ it holds that $\sum_{d=1}^D \dualSb_{dt} = (1 - \frac{1}{C}(t+\frac{1}{2})) \cdot W$ due to Lemma~\ref{lemma:kkt-multiplier-bound}. Thus, we conclude the proof with
  \[
    \sum_{t = 0}^{C-1} \sum_{d=1}^D \dualSb_{dt} = W \sum_{t=0}^{C-1} \biggl(1 - \frac{1}{C}\Bigl(t+\frac{1}{2}\Bigr)\biggr) = \frac{W}{C} \sum_{t=0}^{C-1} \Bigl(t+\frac{1}{2}\Bigr) = \frac{1}{2} \cdot W \cdot C = \frac{1}{2} \cdot \alg \ . 
  \]
  This concludes the proof of the lemma.
\end{proof}

\begin{lemma}\label{lemma:common-c-dual-feasibility}
 The solution $(\dualSa_j)_j$ and $(\dualSb_{dt})_{d,t}$ is feasible for $(\mathrm{DLP}(1))$.
\end{lemma}

\begin{proof}
 The solution is clearly non-negative. We now verify the dual constraint.
 Fix a job $j$ and a time $t \in \{0,\ldots,C-1\}$. Then, 
 \begin{align*}
     \frac{ \dualSa_j }{p_j} - \Bigl(t + \frac{1}{2}\Bigr) \cdot \frac{w_j}{p_j}
     &=  \biggl(C - \Bigl(t+ \frac{1}{2} \Bigr) \biggr) \cdot \frac{w_j}{p_j} 
     =  \biggl(C - \Bigl(t+ \frac{1}{2} \Bigr) \biggr) \cdot \frac{w_j}{y_{j}}  \cdot \frac{y_{j}}{p_j} \\  
     &= \biggl(C - \Bigl(t+ \frac{1}{2} \Bigr) \biggr) \cdot \bigg(\sum_{d=1}^D b_{dj} \lagrangevar_{d} \bigg) \cdot \frac{y_{j}}{p_j}   
     =    \biggl(1 - \frac{1}{C}\Bigl(t+ \frac{1}{2} \Bigr) \biggr) \cdot \bigg( \sum_{d=1}^D b_{dj} \lagrangevar_{d} \bigg) \\
     &= \sum_{d=1}^D b_{dj} \cdot \biggl(1 - \frac{1}{C}\Bigl(t+ \frac{1}{2} \Bigr) \biggr) \cdot \lagrangevar_{d} =  \sum_{d=1}^D b_{dj} \dualSb_{dt} \ .
 \end{align*}
 The third equality uses~\eqref{psp-kkt1}. 
 To see the fourth equality, note that $C \cdot y_j = p_j$.
 This shows that the dual assignment satisfies the constraint of $(\mathrm{DLP}(1))$ with equality.
\end{proof}

\begin{theorem}\label{theorem:pf-common-c}
 \pf has a competitive ratio equal to $2$ for minimizing the total weighted completion time for \psp with uniform release dates whenever it completes all jobs at the same time.
\end{theorem}

\begin{proof}
Weak duality and Lemma~\ref{lemma:common-c-dual-objective} and~\ref{lemma:common-c-dual-feasibility} imply
\begin{align*}
    \opt \geq  \sum_{j \in J} \dualSa_j - \sum_{t = 0}^{C-1} \sum_{d=1}^D \dualSb_{dt} = \frac{1}{2} \cdot \alg \ ,
\end{align*}
which proves the claimed upper bound on the competitive ratio. The lower bound follows by noting that in the deterministic lower bound instance
all jobs also complete at the same time~\cite{MPT94}.
\end{proof}
  
Lemma~\ref{lemma:common-c-dual-objective} and~\ref{lemma:common-c-dual-feasibility} provide the additional insight that \pf's schedule $(y_{j}(t))_{j,t}$ for structured instances is an optimal solution to $(\LP(1))$. This follows from strong duality because $(y_{j}(t))_{j,t}$ is feasible for $(\LP(1))$ and its objective value is also equal to $\frac{1}{2}\alg$ as $\sum_{t=0}^{C-1} (t+\frac{1}{2}) y_j(t) / p_j = \frac{1}{2}C$ for all $j$.

\subsubsection{Combining Optimal Schedules via Superadditivity}

We consider the partition of the processing requirements~$p = \sum_{\ell=1}^L p^{(\ell)}$ introduced in \cref{subsubsec:splitting}. As argued in the proof of Lemma~\ref{lemma:splitting}, in the instances with processing requirements $p^{(\ell)}$ and uniform release dates, all jobs finish at the same time. Let $\opt_{0}(p^{(\ell)})$ be the optimal objective value for these instances. By combining Lemma~\ref{lemma:splitting} and~\ref{theorem:pf-common-c}, we obtain
\[\alg(p) \leq \sum_{\ell=1}^{L} \biggl(2 \cdot \opt_{0}(p^{(\ell)}) + (E_{\ell+1} - E_\ell) \mathrlap{{}\cdot{}} \sum_{j \in U(E_\ell) \setminus J(E_\ell)} w_j\biggr) \ . \]
The definition of $\alpha$-superadditivity implies that this is at most
\[
 2\alpha \cdot  \opt_{0}(p) + \sum_{j=1}^n w_j \sum_{\ell: E_\ell < r_j} (E_{\ell+1} - E_\ell) 
 = 2\alpha \cdot \opt_{0}(p) + \sum_{j=1}^n w_j r_j \le (2\alpha+1) \cdot \opt(p) \ ,
\]
where the final inequality uses the trivial bound $r_j \le C_j^{\opt}$ for all $j \in [n]$ and that the problem with uniform release dates is a relaxation of the general \psp. This proves the bound for the general case. Note that for uniform release dates, we can assume that $\sum_{j=1}^n w_j r_j = 0$, so that we can bound the objective of $\pf$ by $2\alpha \cdot \opt$. This concludes the proof of Theorem~\ref{thm:superadditive}. 

\subsection{Applications in Machine Scheduling} \label{subsec:superadditive-applications}

In this \lcnamecref{subsec:superadditive-applications}, we apply the analysis framework from \cref{sec:framework} 
to several important machine scheduling environments.

\subsubsection{Unrelated Machines}
In order to show superadditivity 
for preemptive unrelated machine scheduling, denoted as \abc{R}{\pmtn}{\sum w_j C_j}, we fall back on the non-preemptive (and hence non-migratory) variant. In this problem, denoted as \abc{R}{}{\sum w_j C_j}, every job must be processed uninterruptedly on a complete machine, i.e., it processes at a rate equal to the speed of the machine for this job.
For fixed weights and speeds, and for processing requirements~$p$ let $\optnp(p)$ denote the optimal objective value of a non-preemptive schedule for 
release dates~$0$.
Note that this non-preemptive problem is not a special case of \psp. However, it is straightforward to extend the definition of $\alpha$-superadditivity~to~it. 

\begin{lemma}\label{lemma:non-preemptive-superadditive}
  The problem \abc{R}{}{\sum w_j C_j} is 1-superadditive.
\end{lemma}
\begin{proof}
  We can model any schedule for the non-preemptive problem using binary variables~$x_{ijk}$ which indicate whether job $j$ is being scheduled in the $k$th last position on machine $i$. The objective value of this schedule for processing requirements $p_1,\ldots,p_n$ is then equal to
  \[
      \sum_{i = 1}^m \sum_{k=1}^n \sum_{j=1}^n \frac{p_j}{s_{ij}}  x_{ijk}  \sum_{j'=1}^n w_{j'} \sum_{k'=1}^k x_{ij'k'} \ ,
  \]
  subject to the constraints that
  $
      \sum_{i=1}^m \sum_{k=1}^n x_{ijk} = 1    
  $
  for every job $j \in [n]$ and
  $
     \sum_{j=1}^n x_{ijk} \leq 1    
  $
  for every position $k \in [n]$ and machine $i \in [m]$.

  We now prove that the problem is 1-superadditive.
  Let $(x_{ijk})_{i,j,k}$ model an optimal schedule for processing requirements $p$, and let $p = \sum_{\ell=1}^L p^{(\ell)}$ be an arbitrary partition of the processing requirements. Thus,
  \begin{align*}
    \optnp(p) &= 
    \sum_{i = 1}^m \sum_{k=1}^n \sum_{j=1}^n  \frac{p_j}{s_{ij}}  x_{ijk}  \sum_{j'=1}^n w_{j'} \sum_{k'=1}^k x_{ij'k'} \\
    &= \sum_{\ell=1}^L \bigg( \sum_{i = 1}^m \sum_{k=1}^n \sum_{j=1}^n  \frac{p^{(\ell)}_j}{s_{ij}}  x_{ijk}  \sum_{j'=1}^n w_{j'} \sum_{k'=1}^k x_{ij'k'} \bigg) \geq \sum_{\ell=1}^L \optnp(p^{(\ell)}) \ ,
  \end{align*}
  where the inequality holds because for every $1 \leq \ell \leq L$, $(x_{ijk})_{i,j,k}$ is a feasible solution to the instance with processing requirements $p^{(\ell)}$.
\end{proof}

\begin{lemma}\label{lem:1.81-superadditive}
  The problem \abc{R}{\pmtn}{\sum w_j C_j} is 1.81-superadditive.
\end{lemma}

\begin{proof} 
Clearly, the preemptive problem is a relaxation of the non-preemptive problem, i.e., $\opt_{0}(p^{(\ell)}) \le \optnp(p^{(\ell)})$ for all $\ell \in [L]$.
Further, we apply a known bound on the power of preemption in unrelated machine scheduling. \Textcite[Corollary~3]{Sitters17} proved that $\optnp(p) \le 1.81 \cdot \opt_{0}(p)$ for any instance~$p$.
By combining these bounds with Lemma~\ref{lemma:non-preemptive-superadditive}, we obtain
\[
 \sum_{\ell=1}^L \opt_{0}(p^{(\ell)}) \le \sum_{\ell=1}^L \optnp(p^{(\ell)}) \le \optnp(p) \le 1.81 \cdot \opt_{0}(p)\ , 
\]
showing that the preemptive problem is 1.81-superadditive. 
\end{proof}

Theorem~\ref{thm:superadditive} implies that \pf is $4.62$-competitive for the online problem with jobs arriving at their release dates
and $3.62$-competitive for the problem without release dates, proving Theorem~\ref{theorem:unrelated}.

As discussed in \Cref{sec:preliminaries}, we loose a factor of $1+\varepsilon$ in the competitive ratio when requiring polynomial running time for \pf. However, the precise constant proven by \textcite{Sitters17} is actually 
the root of $8x^3-11x^2-4x-4$ which is at most $1.806$. Therefore, the stated bounds still hold.

\subsubsection{Restricted Assignment}
Sitters~\cite{Sitters05} showed that for any instance of \abc{R}{\pmtn, s_{ij} \in \{0,1\}}{\sum C_j} there exists an optimal solution which is non-preemptive. Therefore, Lemma~\ref{lemma:non-preemptive-superadditive} implies that the problem is a 1-superadditive \psp. By Theorem~\ref{thm:superadditive}, the \pf algorithm is thus $2$-competitive for uniform release dates and $3$-competitive in general. 
As discussed in \Cref{subsec:monotone-applications}, \pf can be executed in strongly polynomial time in this setting.
This proves the part of Theorem~\ref{theorem:2-competitive} on restricted assignment.

\subsubsection{Related Machines}
We now consider the special case of \psp where all jobs are available at the beginning and have unit weight, and the speed of a machine~$i$ when processing a job~$j$ is independent of the job, i.e., $s_{ij} = s_i$ for all $j \in [n]$. This is noted as \abc{Q}{\pmtn}{\sum C_j} in the 3-field notation. We assume w.l.o.g.\ that $s_1 \geq \ldots \geq s_m$, and $m \ge n$. The latter can be ensured by adding $n-m$ speed-zero machines.

\begin{lemma}
  The problem \abc{Q}{\pmtn}{\sum C_j} is 1-superadditive.
\end{lemma}

\begin{proof}
  We first characterize the optimal objective value in terms of processing requirements $p_1 \leq \ldots \leq p_n$. We use that the \pspt algorithm computes an optimal solution for this problem~\cite{gonzalez1977optimal,LLLRK84}. This algorithm runs at any time~$t$ the $k$ shortest unfinished jobs on the $k$ fastest machines, for any $1 \leq k \leq |J(t)|$.

  It is not hard to see that \pspt finishes the jobs in the order of their indices. Moreover, the resulting completion times~$C_j$ satisfy for every $1 \leq k \leq n$ that (using $C_0 \coloneq 0$)
  \begin{equation*}
    \sum_{j=1}^k s_{k-j+1} \cdot (C_j - C_{j-1}) = p_k \ ,
  \end{equation*}
  which yields by summation for every $1 \leq k \leq n$ that
  \begin{equation*}
   \sum_{j=1}^k s_{k-j+1} \cdot C_j = \sum_{j=1}^k p_j \ . \label{eq:pspt-structure}  
  \end{equation*}
  This equation can be written as
  \begin{equation}
   \begin{pmatrix} 
    s_1    &         &        & 0 \\
    s_2    & s_1                  \\
    \vdots &         & \ddots     \\
    s_n    & s_{n-1} & \ldots & s_1
   \end{pmatrix}
   \begin{pmatrix}
    C_1 \\ C_2 \\ \vdots \\ C_n
   \end{pmatrix}
   =
   \begin{pmatrix}
    1      &   &        & 0 \\
    1      & 1 &            \\
    \vdots &   & \ddots     \\
    1      & 1 & \ldots & 1
   \end{pmatrix}
   \begin{pmatrix}
    p_1 \\ p_2 \\ \vdots \\ p_n
   \end{pmatrix}\ .
   \label{eq:matrix-related}
  \end{equation}
  Let $\lambda$ be the solution to 
  \[\bone^\top = \lambda^\top 
  \begin{pmatrix} 
   s_1    &         &        & 0 \\
   s_2    & s_1                  \\
   \vdots &         & \ddots     \\
   s_n    & s_{n-1} & \ldots & s_1
  \end{pmatrix}\ .\]
  Then, for every $i=n,\dotsc,1$, we have
  $s_1 \lambda_i = 1 - \sum_{k=i+1}^n s_{k-i+1} \lambda_k \ge 1 - \sum_{k=i+1}^n s_{k-i} \lambda_k = 0$. Therefore, $\mu$ defined by $\mu_k \coloneq \sum_{i=k}^n \lambda_i$ for every $k \in [n]$ is ordered $\mu_1 \geq \ldots \geq \mu_n$. Using this notation, by multiplying \eqref{eq:matrix-related} from the left with $\lambda^\top$, we obtain
  \[
      \opt(p) = \sum_{k=1}^n C_k = \sum_{k=1}^n \mu_k \cdot p_k \ .
  \]
  
  We now prove 1-superadditivity. 
  Let $p = \sum_{\ell=1}^L p^{(\ell)}$ be an arbitrary partition of the processing requirements.
  For every $\ell \in [L]$ let $\sigma_\ell \colon [n] \to [n]$ be a permutation such that $p^{(\ell)}_{\sigma_{\ell}(1)} \leq \ldots \leq p^{(\ell)}_{\sigma_{\ell}(n)}$. Since $\mu_1 \geq \ldots \geq \mu_n$, we have that
  \[
    \opt(p^{(\ell)}) = \sum_{k=1}^n \mu_k \cdot p^{(\ell)}_{\sigma_{\ell}(k)} \leq \sum_{k=1}^n \mu_k \cdot p^{(\ell)}_{k} \ .
  \]
  Therefore, we can conclude that the problem is 1-superadditive as follows:
  \[
    \sum_{\ell=1}^L \opt(p^{(\ell)}) \leq \sum_{\ell=1}^L \sum_{k=1}^n \mu_k \cdot p^{(\ell)}_{k} = \sum_{k=1}^n \mu_k  \sum_{\ell=1}^L p^{(\ell)}_{k} = \sum_{k=1}^n \mu_k \cdot p_k = \opt(p) \ . 
  \]
  This concludes the proof.
\end{proof}

By Theorem~\ref{thm:superadditive}, the \pf algorithm is thus $2$-competitive for uniform release dates, and $3$-competitive in general. 
Moreover, as discussed in \Cref{subsec:monotone-applications}, it can be executed in strongly polynomial time.
This proves the part of Theorem~\ref{theorem:2-competitive} on related machines.

\subsubsection{Parallel Identical Machines}
For parallel identical machines, denoted by \abc{P}{\pmtn}{\sum w_j C_j}, we recover a result by \textcite{BBEM12}, who prove that their algorithm WDEQ, which we show corresponds to \pf in this setting (cf.\ \Cref{app:pf-combinatorial}), is 2-competitive.
For completeness, we argue that this also falls out of our analysis.
This follows from a nowadays folkloric result by \textcite{McN59} who showed that for any instance of 
\abc{P}{\pmtn}{\sum w_j C_j} there exists an optimal non-preemptive solution. Therefore, Lemma~\ref{lemma:non-preemptive-superadditive} implies that the problem is 1-superadditive, and, thus, Theorem~\ref{thm:superadditive} gives that \pf is 2-competitive.

\section{Concluding Remarks}

Our techniques led to substantially improved competitive ratios for \psp and machine scheduling under both online arrival and non-clairvoyance. We closed the gap to the lower bound of $2$ for minimizing the total completion time for related machines and for restricted assignment. In fact, this holds for any scheduling instance for which an optimal solution does not require job migration.

For weighted jobs and unrelated machines, there remains the gap between~$3.62$~(Theorem~\ref{theorem:unrelated}) and the lower bound of~$2$. We remark that improving the bound on the power of preemption of $1.81$ \cite{Sitters17}
directly implies an improved  bound on the competitive ratio by Lemma~\ref{lem:1.81-superadditive} and~\ref{thm:superadditive}. However, this cannot completely close the gap to $2$, as there is a known lower bound of $1.39$ on the power of preemption~\cite{EpsteinLSS17}. 
In contrast,
our analysis established a tight competitive ratio for related machines, despite the fact
that an optimal solution migrates jobs.
We conjecture that the competitive ratio of \pf is exactly~$2$ on unrelated machines.

Another well-studied objective is to minimize the total weighted flow time. \textcite{ImKM18} showed that \pf is $\cO(1/\varepsilon^2)$-competitive for \monpsp if the given resource capacity per time unit is increased by a factor $(e+\varepsilon)$, for any $\varepsilon \in (0,1/2)$. This is called speed augmentation. For the special case of  related machines and equal weights,
\textcite{Gupta2010} showed that \pf is $(2+\varepsilon)$-speed $\cO(1)$-competitive, which is tight for \pf even on a single machine~\cite{KalyanasundaramP00}. Our improved understanding of \pf (cf.\ \cref{alg:pf-related-combinatorial}) might help to close this gap also for the weighted setting.
We note that a different and arguably less natural generalization of RR to (un-)related machines achieves this tight result~\cite{ImKMP14}. Moreover, it is open whether a combination between \pf and a technique for \emph{scalable} algorithms ($(1+\varepsilon)$-speed $\cO(1)$-competitive algorithms), such as LAPS~\cite{EdmondsP12} or its smoothed variant~\cite{ImKMP14}, yields a scalable algorithm for unrelated machines or \monpsp. While for unrelated machines a tailored and rather complex algorithm achieves this~\cite{ImKMP14}, it is known that for the general PSP no algorithm can be constant competitive with constant speed augmentation~\cite{ImKM18}.

An intermediate model between clairvoyant and non-clairvoyant (online) scheduling is stochastic (online) scheduling, where job size distributions become known upon job release. Interestingly, our bounds for the non-clairvoyant \pf achieve the best known performance guarantees for polynomial-time preemptive stochastic scheduling policies. This raises the question of whether 
policies could benefit from stochastic information, which is known to be true for a single machine~\abc{1}{\pmtn}{\sum w_j C_j}~\cite{Sev74}.

\section*{Acknowledgements}
We thank Jugal Garg and László Végh for %
their help with our questions on matching markets and giving us further pointers to literature.
We also thank Nick Stelke for discussions on the special case of minimizing the total completion time on related machines.

\section*{Appendices}
\appendix

\section{Improved and Simplified Analysis of \pf for \psp}\label{app:pf}

We give a simplified and improved variant of the analysis of \pf for \psp of \textcite{ImKM18}.

\begin{theorem}\label{pf:thm:general-pf}
    \pf has a competitive ratio of at most $27$ for minimizing the total weighted completion time of \psp.
\end{theorem}

Fix an instance and \pf's schedule.
Let~$\kappa \geq 1$ and~$0 < \lambda < 1$ be constants, which we fix later.
In the following, we assume by scaling that all weights are integers.

For every time~$t$ let $\zeta(t)$ be the weighted $\lambda$-quantile of the values~$y_j(t)/p_j$, $j \in U(t)$, w.r.t.\ the weights~$w_j$. More formally, if $Z_t$ denotes the sorted (ascending) list of length~$W(t)$ composed of $w_j$ copies of~$y_{j}(t)/p_j$ for every~$j \in U(t)$, then $\zeta(t)$ is the value at the index~$\ceil{\lambda W(t)}$ in~$Z_t$.
Let $\dualSa_{jt} \coloneqq w_j \cdot \ind\bigl[y_j(t)/p_j \leq \zeta(t)\bigr]$ for $t \ge 0$ and $j \in U(t)$, where $\ind[\varphi]$ is the indicator variable of the expression~$\varphi$. Further, for $t \ge 0$ let $\lagrangevar_d(t)$ be the optimal Lagrange multiplier corresponding to the constraint~$d \in [D]$ of $(\CP(J(t)))$. We consider the following dual solution:
\begin{itemize}
    \item $\dualSa_j \coloneqq \sum_{t=0}^{C_j} \dualSa_{jt}$ for every job $j \in J$,
    \item $\dualSb_{dt} \coloneqq \frac{1}{\kappa} \sum_{t' \geq t}   \zeta(t') \cdot \lagrangevar_{d}(t')$ for every $d \in [D]$ and time $t \geq 0$.
\end{itemize}

We first show in the following lemma that the objective value of $(\DLP(\kappa))$ for this assignment upper bounds a constant fraction of the algorithm's objective value.

\begin{lemma}\label{lemma:pf-dual-objective}
    It holds that $\bigl(\lambda - \frac{1}{(1-\lambda)\kappa}\bigr) \alg \leq \sum_{j \in J} \dualSa_j - \sum_{d=1}^D \sum_{t \geq 0} \dualSb_{dt}$.
\end{lemma}

The lemma follows from the following two statements.

\begin{lemma}\label{lemma:pf-dual-obj1}
    It holds that $\sum_{j \in J} \dualSa_j \geq \lambda \cdot \alg$.
\end{lemma}

\begin{proof}
  Consider a time~$t$. Observe that~$\sum_{j \in U(t)} \dualSa_{jt}$ contains the total weight of jobs~$j$ that satisfy~$y_j(t)/p_j \leq \zeta(t)$. By the definition of $\zeta(t)$, we conclude that this is at least~$\lambda \cdot W(t)$, i.e.,~$\sum_{j \in U(t)} \dualSa_{jt} \geq \lambda \cdot W(t)$.
  The statement then follows by summation over time.
\end{proof}

\begin{lemma}\label{lemma:pf-dual-obj2}
    At any time $t$, it holds that $\sum_{d=1}^D \dualSb_{dt} \leq \frac{1}{(1-\lambda) \kappa} W(t)$.
\end{lemma}
\begin{proof}
  Fix a time $t$. Observe that for every $t' \geq t$ the definition of $\zeta(t')$ implies that $\sum_{j \in U(t')} w_j \cdot \ind \bigl[ y_j(t')/p_j \geq \zeta(t') \bigr]\ge (1 - \lambda) W(t')$. Thus,
  \begin{equation}  
    \zeta(t') \cdot (1 - \lambda) W(t') \leq  \sum_{j \in U(t')} w_j \cdot \zeta(t') \cdot \ind \Big[ \frac{y_j(t')}{p_j} \geq \zeta(t') \Big] \leq \sum_{j \in U(t')} w_j \cdot \frac{y_j(t')}{p_j} \ . \label{eq:pf-dual-obj2-eq1}
  \end{equation}

  The definition of $\dualSb_{it}$ and Lemma~\ref{lemma:kkt-multiplier-bound} imply
  \begin{align*}
    \sum_{d=1}^D \dualSb_{dt}  = \sum_{d=1}^D \frac{1}{\kappa}\sum_{t' \geq t}   \lagrangevar_d(t') \cdot \zeta(t')  
    =  \frac{1}{\kappa} \sum_{t' \geq t} \zeta(t') \sum_{d=1}^D \lagrangevar_d(t') = \frac 1 \kappa \sum_{t' \ge t} \zeta(t') \sum_{j \in J(t')} w_j \leq \frac{1}{\kappa} \sum_{t' \geq t} \zeta(t') \cdot W(t') \ .
  \end{align*}
 Using~\eqref{eq:pf-dual-obj2-eq1}, we conclude that this is at most
  \begin{align*}
    \frac{1}{\kappa} \sum_{t' \geq t} \zeta(t') W(t')  
    &\leq  \frac{1}{(1 - \lambda)\kappa} \sum_{t' \geq t} \sum_{j \in U(t')} w_j \cdot \frac{y_j(t')}{p_j} \\
    &\leq \frac{1}{(1 - \lambda) \kappa} \sum_{j \in U(t)} w_j \sum_{t' \geq t}  \frac{y_j(t')}{p_j} 
    \leq \frac{1}{(1 - \lambda) \kappa} \sum_{j \in U(t)} w_j \ .
  \end{align*}
  The second inequality follows from $U(t') \subseteq U(t)$ for $t' \geq t$. The third inequality holds because $\sum_{t' \geq t} y_j(t') \leq p_j$ for every job $j$. This concludes the proof of the lemma.
\end{proof}

Next, we show dual feasibility.

\begin{lemma}\label{lemma:pf-dual-feasibility}
 The dual solution $(\dualSa_j)_j$ and $(\dualSb_{it})_{i,t}$ is feasible for $(\DLP(\kappa))$.
\end{lemma}

\begin{proof}
    Fix a job $j$, a machine $i$ and a time $t \geq r_j$.  Then,
    \begin{align*}
        \frac{ \dualSa_j }{p_j} - \frac{w_j}{p_j} \Bigl(t + \frac 1 2\Bigr)
        &\leq \sum_{t' = 0}^{C_j} \frac{w_j}{p_j} \cdot \ind\Big[\frac{y_j(t')}{p_j} \le \zeta(t')\Big] - \frac{w_j}{p_j} t
        \le \sum_{t' = t}^{C_j} \frac{w_j}{p_j} \cdot \ind\Big[ \frac{y_j(t')}{p_j} \leq \zeta(t') \Big] \\
        &=  \sum_{t' = t}^{C_j} \frac{w_j}{y_j(t')} \cdot \frac{y_j(t')}{p_j} \cdot \ind\Big[ \frac{y_j(t')}{p_j} \leq \zeta(t') \Big] \\
        &= \sum_{t' = t}^{C_j} \sum_{d=1}^D b_{dj} \cdot \lagrangevar_d(t') \cdot \frac{y_j(t')}{p_j} \cdot \ind\Big[ \frac{y_j(t')}{p_j} \leq \zeta(t') \Big] \\
        &\leq \sum_{t' = t}^{C_j} \sum_{d=1}^D b_{dj} \cdot \lagrangevar_d(t') \cdot \zeta(t') 
        \leq \sum_{d=1}^D b_{dj} \sum_{t' \geq t} \lagrangevar_d(t') \cdot \zeta(t') 
        = \kappa \sum_{d=1}^D b_{dj} \dualSb_{dt} \ .
    \end{align*}
    The second equality uses~\eqref{psp-kkt1}. This concludes the proof of the lemma.
\end{proof}

Finally, we prove Theorem~\ref{pf:thm:general-pf}.
Weak duality, Lemma~\ref{lemma:pf-dual-objective} and Lemma~\ref{lemma:pf-dual-feasibility} imply
  \begin{align*}
      \kappa \cdot \opt 
      \geq  \sum_{j \in J} \dualSa_{j} - \sum_{d=1}^D \sum_{t \geq 0} \dualSb_{dt} \geq \Bigl(\lambda -  \frac{1}{(1-\lambda)\kappa}\Bigr) \cdot \alg \ .
  \end{align*}
Choosing $\kappa = 9$ and $\lambda = \frac{2}{3}$ implies $\alg \leq 27 \cdot \opt$.

\section{Combinatorial Implementation of \pf for Related Machines}
\label{app:pf-combinatorial}

In this section, we present a combinatorial implementation of \pf for the problem of minimizing the total weighted completion time on uniformly related machines, denoted as \abc{Q}{r_j, \pmtn}{\sum w_j C_j}, which runs in strongly polynomial time. For simplicity, we fix an arbitrary state at time~$t$ and assume that $J(t) = \{1,\ldots,n\}$ and $m=n$. This can be achieved by adding speed-zero machines or removing the $m-n$ slowest machines, as this does not weaken an optimal solution.

We note that combining the results of \textcite{ImKM18} and \textcite{JainV10} already imply that \pf can be implemented in strongly polynomial time (as it models a submodular utility allocation market). Our contribution is to provide a simpler and faster algorithm. Moreover, the description of our algorithm directly implies that,
to the best of our knowledge, all known 2-competitive non-clairvoyant algorithms for uniform release dates are special cases of \pf.

\subsubsection*{A Specialized Convex Program}
Using the representation of the rate polytope for unrelated machine scheduling given in \Cref{sec:preliminaries}, we can see that the convex program~$(\mathrm{\CP}(J(t)))$ for this problem is equal to the following convex program (where the $y$-variables have been eliminated):
{
 \renewcommand{\arraystretch}{1.5}
 \[\begin{array}{rr>{\displaystyle}rcl>{\quad}l}
  (\CP_Q)& \operatorname{maximize} &\multicolumn{3}{l}{\displaystyle \sum_{j=1}^n w_j \log\bigg(\sum_{i=1}^n s_{i} x_{ij} \bigg)} \\
  &\text{subject to} &\sum_{j = 1}^n x_{ij} &\le& 1 & \forall i \in [n] \\
  &&\sum_{i=1}^n x_{ij} &\le& 1 & \forall j \in [n] \\
  &&x_{ij} &\ge& 0 & \forall i \in [n],\, j \in [n]
 \end{array}\]
}%
The KKT conditions~\cite{BV2014} of an optimal solution $(x_{ij})_{i,j}$ of $(\CP_Q)$ with Lagrange multipliers~$(\lagrangevar_i)_i$ and $(\delta_{j})_j$ can be written as follows:
\begin{align}
    -\frac{s_{i} w_j}{\sum_{i'} s_{i'} x_{i'j}} + \eta_{i} + \delta_{j} &\ge 0 \quad \forall i, j \in [n]  \label{pf-related-kkt1}\\
    \eta_{i} \cdot \biggl( 1 - \sum_{j=1}^n x_{ij} \biggr) &= 0 \quad  \forall i \in [n]  \label{pf-related-kkt2} \\
    \delta_{j} \cdot \biggl( 1 - \sum_{i=1}^n x_{ij} \biggr) &= 0 \quad  \forall j \in [n]  \label{pf-related-kkt3} \\
    x_{ij} \cdot \biggl(-\frac{s_{i} w_j}{\sum_{i'} s_{i'} x_{i'j}} + \eta_{i} + \delta_{j}\biggr) &= 0 \quad  \forall i, j \in [n] \label{pf-related-kkt4} \\
    \eta_{i}, \delta_{j} & \geq 0  \quad \forall i, j \in [n] \label{pf-related-kkt5} 
\end{align}

\begin{algorithm}[tb]
  \DontPrintSemicolon
  \caption{Solution algorithm for $(\CP_Q)$}
  \label{alg:pf-related-combinatorial}
  \KwIn{Weights $w_1 \geq \ldots \geq w_n$, machine speeds $s_1 \geq \ldots \geq s_n$.}
  $k \gets 1$ \;
  \While{$k \leq n$}{
      Find the (largest) index $k \leq h \leq n$ that maximizes
  \(
      \frac{\sum_{j=k}^h w_j}{\sum_{i=k}^h s_i}.
  \) \label{pf:line:findlevel} \;
  Compute an allocation of jobs $k,\dotsc,h$ to machines $k,\ldots,h$ such that every such job $j$ receives a rate $y_j = \frac{w_j}{\sum_{j'=k}^h w_{j'}} \cdot \sum_{i=k}^h s_{i}$, as described in \Cref{lemma:compute-allocation}. \label{pf:line:computerate} \;
  $k \gets h+1$ \label{pf:line:nextlevel} \;
  }
\end{algorithm}

\subsubsection*{A Combinatorial Algorithm, Intuition, and Relation to Previous Results}
We now introduce \cref{alg:pf-related-combinatorial}.
We will show below that this algorithm computes an optimal solution of
$(\CP_Q)$ and thus determining the instantaneous resource allocation $x$. 
We prove this by giving an assignment of Lagrangian multipliers which, together with $x$, satisfies the KKT conditions. This is sufficient, because $(\CP_Q)$ has a concave objective function and convex restrictions~\cite{BV2014}.
Executing this algorithm at every event time in the schedule leads to a combinatorial, strongly polynomial implementation of \pf. 

The main idea of \Cref{alg:pf-related-combinatorial} is the following. If there exists an allocation $x_{ij}$ and resulting processing rates~$y_j$, $i,j \in [n]$, such that they fully utilize all machines and all jobs~$j$ have the same ratio $w_j/y_j \eqqcolon \pi$, then we can easily conclude via~\eqref{pf-related-kkt1} that this allocation is optimal for $(\CP_Q)$ by setting $\eta_i \coloneq s_i \pi$ for all $i \in [n]$ and setting all other multipliers to $0$. That means, it is optimal to distribute the processing rates proportional to the weights. However, this is not possible if the distribution of weights is very different to the distribution of speeds.
In this case, we search for the largest prefix of $[h] \subseteq [n]$ (called \emph{level}) such that a proportional allocation is possible for jobs and machines in $[h]$ (cf.~\Cref{pf:line:findlevel}). Then, we remove these machines and jobs from the instance (cf.~\Cref{pf:line:nextlevel}) and repeat.

Consider, for example, identical machines: we would allocate job~$1$ entirely to machine~$1$ if $w_1 / W > 1 / m$, because its fraction of the total weight $W$ is larger than the fraction of the speed of a single machine compared to the total speed $m$. Note that is exactly how \wdeq handles high-weight jobs on identical parallel machines~\cite{BBEM12}. In fact, the instantaneous resource allocation computed by \cref{alg:pf-related-combinatorial} generalizes all known (to the best of our knowledge) non-clairvoyant algorithms for the simpler parallel identical machine environments and uniform release dates: \wdeq~\cite{BBEM12}, \wrr~\cite{KC03}, and \rr~\cite{MPT94}.

We remark that \textcite{GuptaIKMP12} mention another  natural generalization of \wrr to related machines, which allocates a $w_j / W$ fraction of each of the fastest $\floor{W / w_j}$ machines to job $j$.
However, they show that this algorithm is at least $\Omega(\sqrt{\log n})$-competitive for the total weighted completion time objective, whereas 
we prove that \pf is $4$-competitive in that case  (Theorem~\ref{thm:restricted-related-weighted}).

\Textcite{FeldmanMNP08} considered \cref{alg:pf-related-combinatorial} as a mechanism to fractionally allocate $n$ ad slots~$i$ with different click rates~$s_i$ to $n$ advertisers~$j$. Each advertiser~$j$ wishes to maximize their received click rate~$y_j$ and has a budget~$w_j$ they is willing to pay. To this end they can bid an amount $w_j'$. The authors showed that when the slots are allocated according to \cref{alg:pf-related-combinatorial} applied to the bids~$w_j'$ and each advertiser is charged their own bid, then bidding the true budget~$w_j$ is a dominant strategy, i.e., the mechanism is strategy-proof. Therefore, since all bids are collected by the auctioneer, it maximizes the revenue of the auctioneer under strategic behavior of the advertisers. Our result provides further insight about this setting. 
Since the works of \textcite{ImKM18} and \textcite{JainV10} imply that a solution to $(\CP_Q)$ is equal to a market equilibrium, our result shows that the mechanism actually computes a market equilibrium.
That means that it does not have to enforce the allocation, but only needs to set the prices for ad rate in the different slots. For these prices the allocation maximizes each advertisers' click rate subject to their budget constraint, so that it is expected to arise without central coordination.

\paragraph{Proof of Equivalence}
We finally prove the main theorem of this section.

\begin{theorem}\label{thm:pf-related-comb}
    \Cref{alg:pf-related-combinatorial} computes an optimal solution of $(\CP_Q)$.
\end{theorem}

Note that this theorem also implies that the above mentioned non-clairvoyant algorithms are all special cases of \pf.

The remaining section is dedicated to the proof of Theorem~\ref{thm:pf-related-comb}.
Let $L_1,\ldots,L_r$ be the partition of $1,\ldots,n$ produced by the iterations of \Cref{alg:pf-related-combinatorial}. We call each $1\le \ell \le r$ a \emph{level} and denote by
\[
    \pi(L_\ell) = \frac{\sum_{j \in L_\ell} w_j}{\sum_{i \in L_\ell} s_i}
\]
the \emph{price} of level $\ell$. Further, we write $\pi_j = \pi(L_\ell)$ for every $j \in L_\ell$.
Note that for every job~$j$ it holds that $\pi_j = w_j / y_j$.

\begin{lemma}[\cite{FeldmanMNP08}]\label{lemma:compute-allocation}
The allocation in \Cref{pf:line:computerate} always exists, and can be computed efficiently.
\end{lemma}

\begin{proof}
  Consider a level~$L=\{k,\ldots,h\}$ with price $\pi \coloneqq \pi(L)$ computed in \Cref{pf:line:findlevel}. Thus, for every $k \leq h' \leq h$ it holds that $\pi \sum_{i=k}^{h'} s_i \geq \sum_{j=k}^{h'} w_j$. Since $\pi = w_j / y_j$ for every $j \in L$, we conclude that for every $k \leq h' \leq h$ we have
  \[
      \sum_{i=k}^{h'} s_i \geq \sum_{j=k}^{h'} \frac{w_j}{\pi} = \sum_{j=k}^{h'} y_j \ .
  \]
  \Textcite{HorvathLS77} show that these conditions suffice to ensure that jobs of length~$y_k,\dotsc,y_h$ can be preemptively scheduling on machines with speeds~$s_k,\ldots,s_h$ within a makespan of $1$. Moreover, such a schedule can be computed efficiently using the \emph{level algorithm}. 
  This concludes the proof of the lemma.
\end{proof}

\Textcite{GonzalezS78} give an even faster algorithm to compute the allocations in~\cref{pf:line:computerate}. We can derive two immediate observations from the description of \Cref{alg:pf-related-combinatorial}.

\begin{observation}\label{observation:level-allocation}
    If $j \in L_\ell$, then $x_{ij} = 0$ for all $i \in [n] \setminus L_\ell$.
\end{observation}

\begin{observation}\label{lemma:level-monotone}
    It holds that $\pi_1 \geq \dotsc \geq \pi_n$.
\end{observation}

Finally, we prove Theorem~\ref{thm:pf-related-comb}.
\begin{proof}
    Let  $x_{ij}$, $i, j \in [n]$, be the allocation computed by \Cref{alg:pf-related-combinatorial}.
    Since $(\CP_Q)$ has a concave objective function and convex restrictions, we show via the sufficient KKT condition for convex programs that this allocation is an optimal solution to $(\CP_Q)$~\cite{BV2014}.
    To this end, we present Lagrange multipliers such that the KKT conditions are satisfied: for every machine $i \in [n]$ we define
    \[  
        \eta_i := \pi_n s_n + \sum_{k=i}^{n-1} \pi_k \left( s_k - s_{k+1} \right) \ ,
    \]
    and for every job $j \in [n]$ we define 
    \[
        \delta_j := \pi_j s_j - \eta_j \ . 
    \] 
    
    We first verify~\eqref{pf-related-kkt1}. Fix a job $j \in [n]$ and a machine $i \in [n]$. We distinguish two cases.
    \begin{description}
    \item[Case $i \le j$.]
    Using Lemma~\ref{lemma:level-monotone}, we have
    \begin{align}   
        \frac{w_j \cdot s_i}{y_j} = \pi_j s_i 
        = \pi_j \biggl(s_j + \sum_{k=i}^{j-1} (s_k - s_{k+1}) \biggr) 
        &\leq \pi_j s_j + \sum_{k=i}^{j-1} \pi_k(s_k - s_{k+1}) \label{ineq:kkt-related} \\
        &= \pi_js_j + \eta_i - \eta_j
        = \eta_i + \delta_j \ .\notag
    \end{align}
    \item[Case $i > j$.] 
    Using Lemma~\ref{lemma:level-monotone}, we have
    \begin{align*}   
        \frac{w_j \cdot s_i}{y_j} =  \pi_j s_i 
        =  \pi_j\biggl(s_j + \sum_{k=j}^{i-1} (s_{k+1} - s_{k}) \biggr) 
        &\leq  \pi_js_j + \sum_{k=j}^{i-1} \pi_k(s_{k+1} - s_{k})  \\
        &= \pi_js_j - (\eta_j - \eta_i)
         = \eta_i + \delta_j \ .
    \end{align*}
\end{description}

Next, we verify \labelcref{pf-related-kkt2,pf-related-kkt3,pf-related-kkt4}.
Since $n=m$ and by the definition of the algorithm, it is clear that every job and every machine is fully allocated, hence \labelcref{pf-related-kkt2,pf-related-kkt3} follow. For~\eqref{pf-related-kkt4}, note that this follows from Observation~\ref{observation:level-allocation} whenever $i$ and $j$ are in different levels. Otherwise, if $i, j \in L_\ell$, the inequality in \eqref{ineq:kkt-related} holds with equation because in that case $\pi_j = \pi(L_\ell) = \pi_k$ for $k=i,\dotsc,j$.

We finally check the non-negativity \eqref{pf-related-kkt5}. Note that it suffices to show that $0 \leq \eta_i \leq \pi_i s_i$.
To see this, observe that for every machine $i$ we have
\[
    \eta_i = \pi_n s_n + \sum_{k=i}^{n-1} \pi_k \underbrace{ \left( s_k - s_{k+1} \right) }_{\geq 0} \geq 0 \ ,
\]
and, by using Lemma~\ref{lemma:level-monotone}, we have
\[
    \eta_i = \pi_n s_n + \sum_{k=i}^{n-1} \pi_k \left( s_k - s_{k+1} \right) = \pi_i s_i + \sum_{k=i+1}^{n} s_k \underbrace{\left( \pi_{k} - \pi_{k-1} \right)}_{\leq 0} \leq \pi_i s_i \ .
\]
This completes the proof.
\end{proof}

\Cref{alg:pf-related-combinatorial} can be executed in time $O(n^2)$. Since we do this at every release and completion time, the total running time is in $O(n^3)$.

\section{Tight Dual Fitting for Weighted Round-Robin}\label{app:dual-fitting}

In this section we present a tight analysis of the Weighted-Round-Robin algorithm (WRR) via dual fitting.
Let $[n]$ be the set of jobs. We assume that $w_1/p_1 \geq \ldots \geq w_n / p_n$ and that all jobs are available at time $0$.
Then, WRR processes at any time $t$ every unfinished job $j$ at rate $y_j(t) = w_j / \sum_{j' \in U(t)} w_{j'}$.
It is not hard to derive that the completion time of job $j$ in the schedule produced by WRR is equal to
\[
  C_j = \sum_{k=1}^{j-1} p_k + \frac{p_j}{w_j} \sum_{k=j}^n w_k \ .    
\]

We prove the following theorem. 

\begin{theorem}\label{thm:wrr-dual-fitting}
  The objective value of Weighted Round-Robin is equal to twice the optimal objective value of $(\LP(1))$ on a single machine if all jobs are available at time $0$.
\end{theorem}

It is known that the total weighted completion time of a schedule produced by WRR is equal to twice the optimal total weighted mean busy time~\cite{KC03,JagerSSW24}. Since the objective function of $(\LP(1))$ corresponds to the total weighted mean busy time, the statement of Theorem~\ref{thm:wrr-dual-fitting} already follows.

Our contribution is to show how an optimal solution to $(\DLP(1))$ can be expressed in terms of WRR's schedule. This illustrates that we need a substantially more complex dual setup and sufficient knowledge on optimal LP solutions to improve over the factor of 4, which can be achieved by the simple dual setup (cf.\ \Cref{sec:monotone}). In particular, in order to craft an optimal dual solution, it is necessary (due to complementary slackness) to respect that scheduling jobs in WSPT order minimizes the total weighted mean busy time on a single machine~\cite{Goe96}, and, thus, optimally solves $(\LP(1))$.

Let $T = \sum_{j = 1}^n p_j$ be the makespan of any non-idling schedule.
For convenience we restate $(\DLP(1))$ for the special case of single machine scheduling:
{
 \renewcommand{\arraystretch}{1.9}
 \[\begin{array}{rr>{\displaystyle}rcl>{\quad}l}
  (\DLP(1))& \operatorname{maximize} &\multicolumn{3}{l}{\displaystyle\sum_{j \in J} \dualVa_j - \sum_{t = 0}^{T-1} \dualVb_{t}} \\
  &\text{subject to} &\frac{\dualVa_j}{p_j} - \frac{w_j}{p_j} \Bigl(t + \frac{1}{2} \Bigr) &\le& \dualVb_{t} &\forall j \in J,\ \forall t \geq 0 \\
  &&\dualVa_j, \dualVb_{t} &\ge& 0 &\forall j \in J,\ \forall t \geq 0
 \end{array}\]
}%

For the analysis we assume by scaling that all processing requirements and weights are integers.
Define $q_j \coloneq \sum_{k=1}^j p_k$ for every job $j$ and set $q_0 \coloneq 0$. Note that~$q_{j} - q_{j-1} = p_j$. 
We craft the following dual solution:
\begin{itemize}
    \item $\dualSa_j \coloneq w_j C_j$ for every job $j \in [n]$, and
    \item for every $t \in \{0,\ldots,T-1\}$ let $i \in [n]$ such that $q_{i-1} \leq t < q_i$ we set 
    \[
        \dualSb_t \coloneq \frac{w_i}{p_i} \biggl( C_i - \Bigl(t + \frac{1}{2} \Bigr) \biggr) = \sum_{k=i}^n w_k - \frac{w_i}{p_i} \biggl( t + \frac{1}{2} - \sum_{k=1}^{i-1} p_k \biggr) \ .
    \]
\end{itemize}

Note that $\sum_j \dualSa_j = \alg$. Combined with the following lemma, we conclude that the dual objective value of the dual solution equal to $\frac{1}{2} \cdot \alg$.

\begin{lemma}
    It holds that $\sum_{t = 0}^{T-1} \dualSb_t = \frac{1}{2} \cdot \alg$.
\end{lemma}

\begin{proof} 
    \begin{align*}
        \sum_{t=0}^{T-1} \dualSb_t 
        &= \sum_{j=1}^{n} \sum_{t = q_{j-1}}^{q_j - 1}  \frac{w_i}{p_i} \biggl( C_i - \Bigl(t + \frac{1}{2} \Bigr) \biggr) 
        = \sum_{j=1}^{n} \sum_{t = q_{j-1}}^{q_j - 1}  \sum_{k=j}^n w_k - \frac{w_j}{p_j} \biggl(t + \frac{1}{2} - \sum_{k=1}^{j-1} p_k \biggr) \\
        &= \sum_{j=1}^{n} \biggl(p_j \sum_{k=j}^n w_k \biggr) - \frac{w_j}{p_j} \sum_{t = 0}^{p_j - 1} \biggl( t + \frac{1}{2} \biggr) 
        = \sum_{j=1}^{n} \biggl(p_j \sum_{k=j}^n w_k \biggr) - \frac{w_j}{p_j} \cdot \biggl( \frac{p_j (p_j - 1)}{2} + \frac{p_j}{2} \biggr) \\ 
        &= \sum_{j=1}^{n} p_j \biggl(\frac{w_j}{2} + \sum_{k=j+1}^n w_k \biggr) 
        = \frac{1}{2} \sum_{j=1}^{n} p_j \bigg(\sum_{k=j}^n w_k + \sum_{k=j+1}^n w_k \bigg) \\
        &= \frac{1}{2} \sum_{j=1}^{n} p_j \bigg(\sum_{k=j}^n w_k \bigg) + w_j \sum_{k=1}^{j-1} p_k 
        = \frac{1}{2} \sum_{j=1}^{n} w_j  \bigg(\sum_{k=1}^{j-1} p_k + \frac{p_j}{w_j} \sum_{k=j}^n w_k \bigg) = \frac{1}{2} \alg \ .
    \end{align*}
\end{proof}

\begin{lemma}
    The dual solution $(\dualSa_j)_j$ and $(\dualSb_t)_t$ is feasible for $(\DLP(1))$.
\end{lemma}

\begin{proof}
    We prove that for every $j \in [n]$ and $q_{i-1} \leq t \leq q_{i} - 1$ the constraint of $(\DLP(1))$ is satisfied. Let $\tau \in [0,1)$ such that $t + \frac{1}{2} = q_{i-1} +  \tau \cdot p_i$. 
    We distinguish two cases. 
    
    If $j \leq i$, then the left side of the dual constraint is equal to
    \begin{align*}
        \dualSa_j - w_j \Bigl(t + \frac{1}{2} \Bigr) &= w_j \sum_{k=1}^{j-1} p_k + p_j \sum_{k=j}^n w_k - w_j \Bigl(t + \frac{1}{2} \Bigr) \\
        &= w_j \sum_{k=1}^{j-1} p_k + p_j \sum_{k=j}^n w_k - w_j \sum_{k=1}^{i-1} p_k -  \tau \cdot w_j \cdot p_i \\
        &= - w_j \sum_{k=j}^{i-1} p_k + p_j \sum_{k=j}^n w_k -  \tau \cdot w_j \cdot p_i \\
        &= - w_j \sum_{k=j}^{i-1} p_k + p_j \sum_{k=j}^{i-1} w_k + p_j \sum_{k=i}^n w_k  -  \tau \cdot w_j \cdot p_i \ .
    \end{align*}
    Using the fact that $w_j p_k \geq w_k p_j$ for every $k \geq j$, we can bound this expression from above by
    \begin{align*}
        &- w_j \sum_{k=j}^{i-1} p_k + w_j \sum_{k=j}^{i-1} p_k + p_j \sum_{k=i}^n w_k  -  \tau \cdot w_i \cdot p_j  
        = \biggl(  \sum_{k=i}^n w_k  -  \tau \cdot w_i  \biggr) \cdot p_j \\ 
        &= \bigg( \sum_{k=i}^n w_k  - \frac{w_i}{p_i} (\tau \cdot p_i ) \bigg) \cdot p_j  
        = \bigg( \sum_{k=i}^n w_k  - \frac{w_i}{p_i} \Bigl( t + \frac{1}{2} - \sum_{k=1}^{i-1} p_k \Bigr)  \bigg) \cdot p_j = \dualSb_t \cdot p_j \ ,
    \end{align*}
    giving the right side of the constraint.
    
    Similarly, we have for the case $j > i$ that
    \begin{align*}
        \dualSa_j - w_j \Bigl(t + \frac{1}{2} \Bigr) &= w_j \sum_{k=1}^{j-1} p_k + p_j \sum_{k=j}^n w_k - w_j \Bigl(t + \frac{1}{2} \Bigr) \\
        &= w_j \sum_{k=1}^{j-1} p_k + p_j \sum_{k=j}^n w_k - w_j \sum_{k=1}^{i-1} p_k -  \tau \cdot w_j \cdot p_i \\
        &= w_j \sum_{k=i}^{j-1} p_k + p_j \sum_{k=j}^n w_k -  \tau \cdot w_j \cdot p_i \\
        &= w_j \sum_{k=i+1}^{j-1} p_k - p_j \sum_{k=i}^{j-1} w_k + p_j \sum_{k=i}^n w_k + (1- \tau) \cdot w_j \cdot p_i \ .
    \end{align*}
    Using the fact that $w_k p_j \geq w_j p_k$ for every $k < j$, we can bound this expression from above by
    \begin{align*}
        &p_j \sum_{k=i+1}^{j-1} w_k - p_j \sum_{k=i}^{j-1} w_k + p_j \sum_{k=i}^n w_k + (1- \tau) \cdot w_i \cdot p_j 
        = p_j \sum_{k=i}^n w_k - \tau \cdot w_i \cdot p_j \\
        &= \bigg( \sum_{k=i}^n w_k  - \frac{w_i}{p_i} (\tau \cdot p_i )  \bigg) \cdot p_j 
        = \bigg( \sum_{k=i}^n w_k  - \frac{w_i}{p_i} \Bigl( t + \frac{1}{2} - \sum_{k=1}^{i-1} p_k \Bigr)  \bigg) \cdot p_j = \dualSb_t \cdot p_j \ ,
    \end{align*}
    which concludes the statement.
\end{proof}

Theorem~\ref{thm:wrr-dual-fitting} then follows from these two lemmas and  weak LP duality.

\section{Offline Approximation Algorithm for \psp}
\label{app:offline-psp}

\textcite{ImKM18} note that a constant-factor approximation algorithm is possible for \psp. For completeness, we present a $(2+\varepsilon)$-approximation algorithm for \psp for any $\varepsilon > 0$. 
We use a standard technique to randomly round a LP solution, which has been introduced by \textcite{SchulzS97}, and use an interval-indexed LP formulation~\cite{SchulzS02,QueyranneS02}, which can be solved in polynomial time. 

Let $\varepsilon' > 0$ and $\delta > 0$. First, we slightly shift the release date of every job $j$ to $r'_j = r_j + \delta$. 
Note that in an instance with shifted release dates $r'_j$, the optimal objective value is at most $(1+\delta)$ times the optimal objective value in the original instance assuming that $1$ is a lower bound for the first completion time. The latter can be ensured by computing the maximal possible rate per job and then scaling the instance. From now on we only consider this adjusted instance.
Let $T$ be an upper bound on the makespan of the schedule produced by the linear program below, which we can guess and verify in polynomial time using a doubling binary search. 
We now discretize the time between $\delta$ and $T$ into geometric time intervals. 
Let $L$ be the smallest integer such that $\delta(1+\varepsilon')^L \geq T$. 
For every $0 \leq \ell \leq L$ we define $\beta_\ell = \delta (1+\varepsilon')^\ell$ and for every $1 \leq \ell \leq L$ we define the interval $I_\ell = (\beta_{\ell-1}, \beta_{\ell}]$.

We now introduce an interval-indexed variant $(\LP_I)$ of $(\LP(\kappa))$ where the variable $y_{j\ell}$ indicates that job $j$ receives a processing volume equal to $y_{j\ell} \cdot |I_\ell|$ during interval $I_\ell$:
{
 \renewcommand{\arraystretch}{1.5}
 \[\begin{array}{rr>{\displaystyle}rcl>{\quad}l}
  (\LP_I)& \operatorname{minimize} &\multicolumn{3}{l}{\displaystyle\sum_{j \in J} w_{j} \sum_{\ell=1}^L \frac{y_{j\ell} \cdot |I_\ell|}{p_j} \beta_{\ell-1}}  \\
  &\text{subject to} &\sum_{\ell=1}^L y_{j\ell} \cdot |I_\ell| &\ge& p_j & \forall j \in J  \\
  && \sum_{j \in J} b_{dj} \cdot y_{j\ell} &\le& 1 &\forall d \in [D],\ \forall \ell \geq 1 \\
  && y_{j\ell}  &\ge& 0 &\forall  j \in J,\ \forall \ell \geq 1 \\ 
  && y_{j\ell}  &=& 0 &\forall  j \in J,\ \forall \ell \geq 1 \text{ s.t.\ } \beta_{\ell-1} < r'_j
 \end{array}\]
}%
To see that $(\LP_I)$ is a relaxation of \psp, fix an optimal solution. Let $C^{\opt}_j$ denote the completion time of job $j$ in this solution, and let $y^*_{j\ell} |I_\ell|$ be equal to the amount of processing which $j$ receives during interval $I_\ell$. Using $\ell^* = \argmax_\ell y^*_{j\ell} > 0$ gives that $C^{\opt}_j \geq \beta_{\ell^*-1}$, and, thus,
\[
  \sum_{\ell=1}^L \frac{y^*_{j\ell} \cdot |I_\ell|}{p_j} \beta_{\ell-1} \leq \sum_{\ell=1}^L \frac{y^*_{j\ell} \cdot |I_\ell|}{p_j} C^{\opt}_j \leq C^{\opt}_j \ , 
\] 
which implies that the optimal objective value of $(\LP_I)$ is a lower bound on the objective value of an optimal solution.

We execute the following steps to produce a feasible schedule for the \psp.
\begin{enumerate}
  \item Solve $(\LP_I)$ and obtain a solution $(y_{j\ell})_{j,\ell}$ to $(\LP_I)$.
  \item Sample $\alpha \in (0,1)$ with density function $f(x) = 2x$.
  \item Slow down the schedule $(y_{j\ell})_{j,\ell}$ by a factor of $\frac{1}{\alpha}$ and output the resulting schedule. More formally, for every $\ell \geq 1$ schedule
  each job $j \in U(\beta_{\ell-1} / \alpha)$ during time $[\beta_{\ell-1} / \alpha, \beta_{\ell} / \alpha]$ (or until it completes) at rate~$y_{j\ell}$. That is, job $j$ receives a total processing of at most $y_{j\ell} |I_\ell| / \alpha$  during time $[\beta_{\ell-1} / \alpha, \beta_{\ell} / \alpha]$.
\end{enumerate}

For every job $j$, let $C_j$ be the completion of $j$ in the schedule produced by this algorithm, and for every interval $I_{\ell}$, $1 \leq \ell \leq L$, let $\alpha_\ell = \sum_{\ell'=1}^\ell y_{j\ell'} \cdot |I_{\ell'}| / p_j$ denote the fraction of $j$ that has been completed in interval $I_{\ell}$.
Let $\ell_j(\alpha) = \argmin_{\ell \in \mathbb N} \{ \alpha_\ell\, |\, \alpha_\ell \geq \alpha \}$
and $C_j(\alpha) = \beta_{\ell_j(\alpha)-1}$.

\begin{lemma}\label{lemma:slowed-down}
For every job $j$, we have $C_j \leq (1+\varepsilon') \cdot C_j(\alpha) / \alpha$.  
\end{lemma}

\begin{proof}
  By the construction of the final schedule, there is an interval index $1 \leq \ell \leq L$ such that $\beta_{\ell-1}/\alpha < C_j \leq \beta_{\ell} / \alpha$. Since $\beta_{\ell-1}/\alpha < C_j$, it must be due to the construction that $\sum_{\ell = 1}^{\ell-1} y_{j\ell} \cdot |I_\ell| / \alpha < p_j$, and, thus, $\ell_j(\alpha) \geq \ell$. Therefore, 
  \[
    C_j \leq \frac{\beta_\ell}{\alpha} \leq (1+\varepsilon') \frac{\beta_{\ell_j(\alpha)-1}}{\alpha} = (1+\varepsilon') \frac{C_j(\alpha)}{\alpha} \ ,
  \]
  which concludes the proof.
\end{proof}

The following property of $\alpha$-points can also be found in e.g.\ \cite{Goe97,SchulzS97,QueyranneS02}. 

\begin{lemma}\label{lemma:alpha-points}
For every job $j$ it holds that 
\[
  \int_0^1 C_j(x) \; \mathrm{d} x = \sum_{\ell=1}^L \frac{y_{j\ell} \cdot |I_\ell|}{p_j}  \beta_{\ell-1} \ .
\]
\end{lemma}

\begin{proof}
Fix a job $j$.
Note that for every $1 \leq \ell \leq L$ if $\alpha \leq \alpha_{\ell}$, then $C_j(\alpha) \leq  \beta_{\ell-1}$. 
Moreover, it holds that $0 = \alpha_0 \leq \ldots \leq \alpha_L = 1$.
Thus,
\begin{align*}
  \int_0^1 C_j(x) \; \mathrm{d} x = \sum_{\ell=1}^L \int_{\alpha_{\ell-1}}^{\alpha_{\ell}} C_j(x) \; \mathrm{d} x 
  \leq \sum_{\ell =1}^L (\alpha_{\ell} - \alpha_{\ell-1})  \beta_{\ell-1} = \sum_{\ell=1}^L \frac{y_{j\ell} \cdot |I_\ell|}{p_j}  \beta_{\ell-1}  \ ,
\end{align*}
which concludes the proof.
\end{proof}

\begin{theorem}
  There is a polynomial-time randomized $(2+\varepsilon)$-approximation algorithm for \psp for every $\varepsilon > 0$.
\end{theorem}

\begin{proof}
Observe that $L$ is polynomially bounded in the input, and, thus, we can solve $(\LP_I)$ in polynomial time.
Lemma~\ref{lemma:slowed-down} and~\ref{lemma:alpha-points} imply for every job $j$ that
  \[
    \EX[C_j] \leq (1+\varepsilon') \cdot \EX\left[\frac{1}{\alpha} C_j(\alpha)\right] = (1+\varepsilon') \cdot \int_0^1 \frac{C_j(x)}{x} \cdot 2x \; \mathrm{d} x \leq 2(1+\varepsilon')  \cdot \sum_{\ell=1}^L \frac{y_{j\ell} \cdot |I_\ell|}{p_j} \beta_{\ell-1}  \ ,
  \]
which shows that the produced schedule approximates an optimal LP solution within a factor of $2(1+\varepsilon')$ in expectation. Due to the initial shifting of release dates, this gives an overall expected approximation ratio of at most $2(1+\varepsilon')(1+\delta)$.
\end{proof}

\section{Migration is Essential for Non-Clairvoyant Scheduling on Unrelated Machines} \label{app:lb-migratory}

In this section we argue that migration is essential for non-clairvoyant algorithms for unrelated machine scheduling, even in the special case of related machines with uniform weights and release dates. We prove Theorem~\ref{thm:migration-necessary} which states that any non-migratory non-clairvoyant algorithm for minimizing the total completion time of $n$ jobs on related machines, \abc{Q}{\pmtn}{\sum C_j}, has a competitive ratio of at least~$\Omega(\sqrt{n})$. A similar result is known for minimizing the makespan~\cite{EberleGMMZ2024}.

\begin{proof}[Proof of Theorem~\ref{thm:migration-necessary}]
Consider $n$ jobs and $n$ machines, one with speed $\sqrt n$, and all others with speed~$1$. If the algorithm puts any job on a slow machine, then this job turns out to be very long, while all other jobs are negligible small in comparison. Therefore, the competitive ratio of such an algorithm would be at least $\sqrt n$. If the algorithm puts all jobs on the first machine, then they all have processing time~$1$, so that, even in an optimal schedule on that machine, the total completion time would be $\frac{n(n+1)}{2\sqrt n}$. An alternative schedule for all machines would process each job on a different machine, resulting in total completion time $1/\sqrt n + n - 1$. The ratio is thus at least $\frac{n(n+1)}{2(1 + \sqrt n (n - 1))} \in \Omega(\sqrt n)$.
\end{proof}

\section{Unrelated Machine Scheduling is not \pf-Monotone}\label{app:unrelated-not-monotone}

We have shown that the property of \pf-monotonicity is a powerful one that allows for the derivation of very good bounds on the competitive ratio of \pf (see Theorem~\ref{thm:main-monotone}). The fact that several machine scheduling problems in this framework are \pf-monotone, including scheduling on related machines and restricted assignment (see Theorem~\ref{thm:restricted-related-weighted}), raises the question of whether the general unrelated machine scheduling problem also has this property. We answer this question to the negative, even in the absence of weights and release times.

\begin{lemma}
	The unrelated machine scheduling problem \abc{R}{\pmtn}{\sum C_j} is not \pf-monotone.
\end{lemma}

\begin{proof}
The representation of unrelated machines scheduling as \psp is as follows:
\[ 
  \mathcal P \coloneqq \left\{y \in \RR_{\ge 0}^{n}  \bigm| \exists x \in \RR_{\ge 0}^{m\times n}:  y_j = \sum_{i=1}^m s_{ij} x_{ij}\, \forall j \in [n],\ \sum_{j=1}^n x_{ij} \le 1\, \forall i \in [m],\ \sum_{i=1}^m x_{ij} \le 1\, \forall j \in [n]\right\}\ .
\]
Consider the following instance with three machines, unit-weight jobs $J=\{1,2,3\}$ and the following speeds $s_{ij}$ for jobs $j$ on machines $i$:
\begin{align*}
  &s_{11} = 1 \ , \; \;  s_{21} = 0 \ , \; \;  s_{31} = 0 \ , \\ 
  &s_{12} = 2 \ , \; \;  s_{22} = 1 \ , \; \;  s_{32} = 0 \ , \\
  &s_{13} = 0 \ , \; \;  s_{23} = 2 \ , \; \;  s_{33} = 1 \ .
\end{align*}
Let $J' = \{1,2\}$.
The optimal solution to $(\CP(J))$ is $y=(2/3,4/3,4/3)$, and the optimal solution to $(\CP(J'))$ is $y' = (1,1)$. We can observe that \pf's rate for job 2 decreases in the absence of job~3.
\end{proof}

\section{Lower Bound for \pf with Non-Uniform Release Dates}\label{app:release-dates-lb}

\begin{figure}[tb]
  \centering
      \begin{tikzpicture}
          \begin{axis}[
              axis lines = left,
              width=12cm,
              height=6cm,
              legend pos=north east,
              xmin=0, xmax=1.9, ymin=0, ymax=2.2,
              xlabel = time,
              ylabel = remaining volume,
              ytick={1.732 - 1, 1.732, 2},
              yticklabels={$1-r$, $2-r$, 2},
              xtick={2-1.732,1,1.732},
              xticklabels={$r$, 1, $\sqrt{3}$},
          ]               
          
          \addplot [
          domain=0:1.7302, 
          samples=100, 
          color=gray,
          dotted,
          line width=1pt] {2 - 2 / 1.7302 * x};
          \addplot[line width=2pt,color=job1] coordinates {(0,2) (0.2679,2-0.2679)};
          \addplot[line width=2pt,color=job2] coordinates {(0.2679,2-0.2679) (0.2679+0.7320,2-0.2679-1)};
          \addplot[line width=2pt,color=job1] coordinates {(0.2679+0.7320,2-0.2679-1) (1.7320,0)};
          \end{axis}
      \end{tikzpicture}
      \caption{The total completion time of both schedules used in Theorem~\ref{thm:rr-release-dates-js}. The solid lines indicate the completion of jobs in the SRPT schedule; the area below is equal to the total completion time of the SRPT schedule. The colors of the solid line indicate which jobs are processed at which time. The total area below the dotted line is equal to half of the total completion time of RR. Observe that the area under the solid line is (slightly) less than the area under the dotted line.}
      \label{fig:rr-lb-simple}
\end{figure}
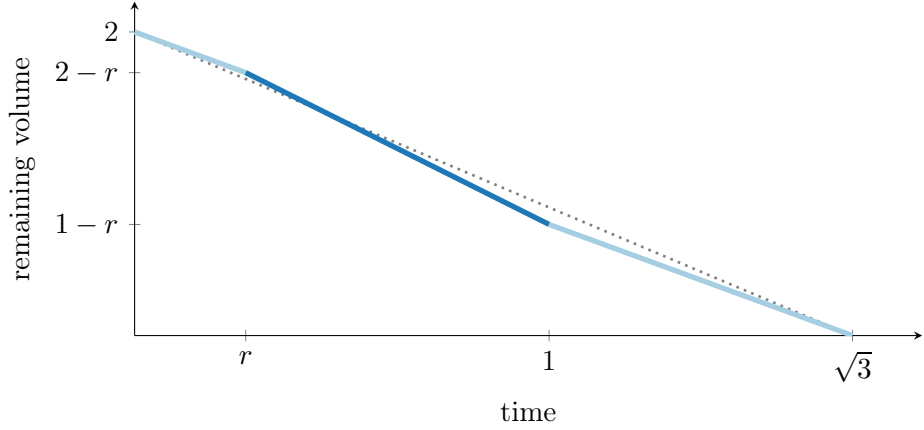

In this section we present a lower bound on the competitive ratio of \pf for \psp which is strictly larger than 2. We prove it for the special case of scheduling on a single machine, where \pf reduces to \rr (cf.\ \Cref{app:pf-combinatorial}).
RR processes at any time $t$ every available job $j \in J(t)$ at rate $1 / |J(t)|$. 
To prove a lower bound on RR's competitive ratio, we compare RR to the optimal solution, which is computed by the Shortest Remaining Processing Time rule (SRPT)~\cite{Schrage68srpt}. This strategy processes at any time the available job of shortest remaining processing time.
We first present a weaker but simpler lower bound, and then give an improved lower bound of at least 2.19.

\begin{theorem}
  \label{thm:rr-release-dates-js}
  The competitive ratio of RR for \abc{1}{r_j,\pmtn}{\sum C_j} is at least $2.074$.
\end{theorem}

\begin{proof}
Consider the following instance.
There are $n$ jobs $J_1$ released at time $0$ with processing time $1$ and $n$ jobs $J_2$ released at time $r = (2 - \sqrt{3})n$ with processing time $\sqrt{3} - 1$. The number of jobs in $J_1$ and $J_2$ is equal to $n$ each, and since we assume that $n \to \infty$, we can w.l.o.g.\ scale $n$ to $1$, that is, the jobs are sand. Note that, in \rr, the jobs in $J_1$ have at time $r$ the same remaining processing time as the jobs in $J_2$. Thus, all jobs complete at time $\sqrt{3}$, and the total completion time of RR is equal to $2\sqrt{3}$. SRPT first sequentially schedules an $r$ fraction of the jobs in $J_1$ until time $r$, then sequentially completes all jobs in $J_2$ and finally completes the remaining $1-r$ volume of $J_1$. In this schedule, the total completion time of the jobs completed until time $r$ is equal to $\frac{1}{2}r^2$, the total completion time of jobs in $J_2$ equal to $r + \frac{1}{2}(\sqrt{3}{}-1)$, and the total completion time of the remaining jobs of $J_1$ equal to $(1-r)(r + \sqrt{3} - 1) + \frac{1}{2}(1-r)^2$. Hence,
\[
  \opt \leq  \frac{1}{2}r^2 + r + \frac{1}{2}(\sqrt{3}{}-1) + (1-r)(r + \sqrt{3} - 1) + \frac{1}{2}(1-r)^2 = 6 - \frac{5}{2}\sqrt{3}\ .
\]
Therefore, the competitive ratio of RR is at least $\frac{2\sqrt{3}}{6 - \frac{5}{2}\sqrt{3}} = \frac{4}{23} (5 + 4 \sqrt{3}) > 2.074$. This construction is depicted in \Cref{fig:rr-lb-simple}.
\end{proof}

\begin{figure}
  \centering
          \begin{tikzpicture}
              \begin{axis}[
                  axis lines = left,
                  width=12cm,
                  height=7.2cm,
                  legend pos=north east,
                  xmin=0, xmax=3.2, ymin=0, ymax=4.5,
                  xlabel = time,
                  ylabel = remaining volume,
                  ytick={4-0.194-0.367109-0.528634-1-0.471366-0.63289,4-0.194-0.367109-0.528634-1-0.471366,4-0.194-0.367109-0.528634-1,4-0.194-0.367109-0.528634,4-0.194-0.367109,4-0.194,4},
                  yticklabels={$1- V_1$,$2-\sum_{i=1}^2 V_i$,$4-\sum_{i=1}^4 V_i$,$4-\sum_{i=1}^3 V_i$,$4-\sum_{i=1}^2 V_i$,,4},
                  xtick={0.1937,0.1937+0.296,0.1937+0.296+0.348,0.1937+0.296+0.348+0.5423,0.1937+0.296+0.348+0.5423+0.3103,0.1937+0.296+0.348+0.5423+0.3103+0.5103,0.1937+0.296+0.348+0.5423+0.3103+0.5103+0.8063},
                  xticklabels={$r_2$, $r_3$, $r_4$, $r_4 + p_4$,, $r_2 + \sum_{i=2}^4 p_i$, $\sum_{i=1}^4 p_i$},
              ]   
              \addplot [
                  domain=0:3, 
                  samples=100, 
                  color=gray,
                  dotted,
                  line width=1pt] {4 - 4 / 3 * x};            
              \addplot[line width=2pt,color=job1] coordinates {(0,4) (0.1937,4-0.194)};
              \addplot[line width=2pt,color=job2] coordinates {(0.1937,4-0.194) (0.1937+0.296,4-0.194-0.367109)};
              \addplot[line width=2pt,color=job3] coordinates {(0.1937+0.296,4-0.194-0.367109) (0.1937+0.296+0.348,4-0.194-0.367109-0.528634)};
              \addplot[line width=2pt,color=job4] coordinates {(0.1937+0.296+0.348,4-0.194-0.367109-0.528634) (0.1937+0.296+0.348+0.5423,4-0.194-0.367109-0.528634-1)};
              \addplot[line width=2pt,color=job3] coordinates {(0.1937+0.296+0.348+0.5423,4-0.194-0.367109-0.528634-1) (0.1937+0.296+0.348+0.5423+0.3103,4-0.194-0.367109-0.528634-1-0.471366)};
              \addplot[line width=2pt,color=job2] coordinates {(0.1937+0.296+0.348+0.5423+0.3103,4-0.194-0.367109-0.528634-1-0.471366) (0.1937+0.296+0.348+0.5423+0.3103+0.5103,4-0.194-0.367109-0.528634-1-0.471366-0.632891)};
              \addplot[line width=2pt,color=job1] coordinates {(0.1937+0.296+0.348+0.5423+0.3103+0.5103,4-0.194-0.367109-0.528634-1-0.471366-0.632891) (0.1937+0.296+0.348+0.5423+0.3103+0.5103+0.8063,4-0.194-0.367109-0.528634-1-0.471366-0.632891-0.8063)};
              \end{axis}
          \end{tikzpicture}
          \caption{
              The total completion time of both schedules in the improved construction with $k=4$. Observe that ratio between the area under the solid line and the area under the dotted line is smaller than the corresponding ratio in \Cref{fig:rr-lb-simple}.
              Here, $p_1=1$, $p_2 \approx 0.80632$, $p_3 \approx 0.65835$, and $p_4 \approx 0.54231$. 
          }
          \label{fig:rr-lb-advanced}
\end{figure}
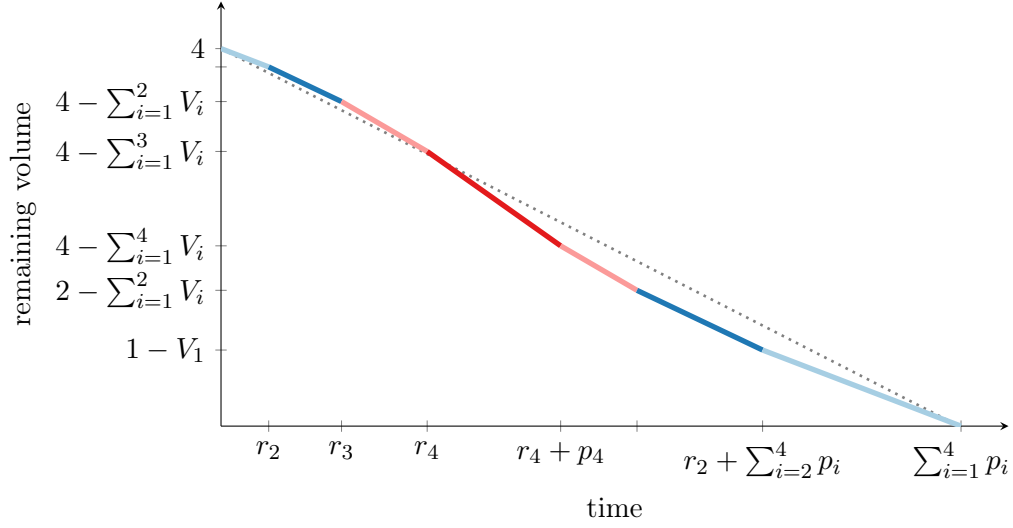
\begin{table}[tb]
  \caption{The (approximate) processing times used for the proof of Theorem~\ref{thm:rr-release-dates-improved}.}\label{table:rr-lb-processing-times}
  \begin{minipage}{0.19\textwidth}
      \centering
      \begin{tabular}{lr}
          \toprule
         $j$ & $p_j$ \\
          \midrule
          1 & 1 \\
          2 & 0.95160 \\
          3 & 0.90702 \\
          4 & 0.86581 \\
          5 & 0.82759 \\
          6 & 0.79203 \\
          \bottomrule
      \end{tabular}
  \end{minipage}
  \begin{minipage}{0.19\textwidth}
      \centering
      \begin{tabular}{lr}
          \toprule
         $j$ & $p_j$ \\
          \midrule
          7 & 0.75885 \\
          8 & 0.72782 \\
          9 & 0.69872 \\
          10 & 0.67137 \\
          11 & 0.64561 \\
          12 & 0.62131 \\
          \bottomrule
      \end{tabular}
  \end{minipage}
  \begin{minipage}{0.19\textwidth}
      \centering
      \begin{tabular}{lr}
          \toprule
         $j$ & $p_j$ \\
          \midrule
          13 & 0.59832 \\
          14 & 0.57656 \\
          15 & 0.55590 \\
          16 & 0.53628 \\
          17 & 0.51760 \\
          18 & 0.49980 \\
          \bottomrule
      \end{tabular}
  \end{minipage}
  \begin{minipage}{0.19\textwidth}
      \centering
      \begin{tabular}{lr}
          \toprule
         $j$ & $p_j$ \\
          \midrule
          19 & 0.48282 \\
          20 & 0.46659 \\
          21 & 0.45106 \\
          22 & 0.43619 \\
          23 & 0.42194 \\
          24 & 0.40825 \\
          \bottomrule
      \end{tabular}
  \end{minipage}
  \begin{minipage}{0.19\textwidth}
      \centering
      \begin{tabular}{lr}
          \toprule
         $j$ & $p_j$ \\
          \midrule
          25 & 0.39510 \\
          26 & 0.38245 \\
          27 & 0.37027 \\
          28 & 0.35854 \\
          29 & 0.34722 \\
          30 & 0.33630 \\
          \bottomrule
      \end{tabular}
  \end{minipage}
\end{table}

We strengthen this example by adding more release dates, which improves the lower bound on the competitive ratio. We consider an instance with $k$ release dates, and at every release date (including $r_1 = 0$) the same number of jobs are released. We again assume that the jobs are sand, hence scale their number to $1$.
Let $J_1,\ldots,J_k$ be those sets, and let $p_1 > \dotsb > p_k$ be their processing times. Our goal is to construct an instance in which all jobs complete at the same time in the schedule of RR with a total completion time equal to $k \sum_{i=1}^k p_i$. Thus, given values for $p_1,\ldots,p_k$, we can observe that the release dates $r_2,\ldots,r_k$ must satisfy that at every time $r_i$ all jobs of $J_1,\ldots,J_i$ have the same remaining processing time. This requires that
\[
r_i = r_{i-1} + p_{i-1} - p_i
\]
for all $i = 2,\ldots,k$.

We now consider the solution of SRPT for this instance. Let $V_j$ be the fraction of jobs $J_j$ that is being processed before the release of the jobs~$J_{j+1}$ in the SRPT schedule, and let $V_k = 1$. 
Note that $V_j = (r_{j+1} - r_j) / p_j$ for all $j=1,\ldots,k-1$.
Thus, the total completion time of SRPT can be expressed as follows (using $r_{k+1} = r_k + p_k$):
\[
  \sum_{j=1}^k (r_{j+1} - r_j) \biggl( k - \frac{1}{2}V_j - \sum_{i=1}^{j-1} V_i \biggr) + \sum_{j=1}^{k-1} (p_j - (r_{j+1} - r_j)) \biggl( \frac{1}{2}(1-V_j) + \sum_{i=1}^{j-1} (1-V_i) \biggr) \ .
\]

We numerically maximized the ratio between RR's objective and this expression over all $1 = p_1 > \ldots > p_k$ for $k=30$ and derived the following improved lower bound. 
The processing times are presented in \Cref{table:rr-lb-processing-times} and the construction is illustrated in \Cref{fig:rr-lb-advanced}. 

\begin{theorem}
  \label{thm:rr-release-dates-improved}
  The competitive ratio of RR for \abc{1}{r_j,\pmtn}{\sum C_j} is at least $2.1906$.
  \end{theorem}

\printbibliography

\end{document}